%% file: manuscript.tex
\documentclass[11pt]{article}
\usepackage[letterpaper,margin=1in]{geometry}
\newcommand\blfootnotea[1]{%
  \begingroup
  \renewcommand\thefootnote{}\footnote{#1}%
  \endgroup
}
\usepackage{dsfont}

\usepackage[letterpaper,margin=1in]{geometry}
\usepackage[utf8]{inputenc} %
\usepackage[T1]{fontenc}    %
\usepackage{hyperref}       %
\usepackage{url}            %
\usepackage{booktabs}       %
\usepackage{amsfonts}       %
\usepackage{nicefrac}       %
\usepackage{microtype}      %
\usepackage{xcolor}

\input{macros}

\def\colorful{0}

\ifnum\colorful=1
\usepackage[normalem]{ulem}

  \newcommand{\anote}[1]{\footnote{{\bf [Ankit: {#1}\bf ]}}}
  \newcommand{\inote}[1]{\footnote{{\bf [[Ilias: {#1}\bf ]] }}}
  \newcommand{\cnote}[1]{\footnote{{\bf [[Chao: {#1}\bf ]] }}}

\newcommand{\todo}[1]{{\textbf{ [{\red{Todo}}: {#1}]}}}
\else

\newcommand{\light}[1]{{}}
\newcommand{\anote}[1]{}
\newcommand{\inote}[1]{}
\newcommand{\cnote}[1]{}
\newcommand{\todo}[1]{}
\fi

\title{Information-Computation Tradeoffs 
for Noiseless Linear Regression with Oblivious Contamination\blfootnotea{Authors are listed in alphabetical order.}
}

\author{
Ilias Diakonikolas\thanks{Supported by NSF Medium Award CCF-2107079 and an H.I. Romnes Faculty Fellowship.}\\
University of Wisconsin, Madison \\
{\tt ilias@cs.wisc.edu} \\
\and 
Chao Gao\thanks{Supported by NSF Grants ECCS-2216912 and DMS-2310769 and an Alfred Sloan fellowship.}\\
University of Chicago \\
{\tt chaogao@uchicago.edu}
\and 
Daniel M.\ Kane\thanks{Supported by NSF Medium Award CCF-2107547.}\\
University of California, San Diego \\
{\tt dakane@ucsd.edu} \\
\and
John Lafferty\\
Yale University \\
{\tt john.lafferty@yale.edu}
\AND
Ankit Pensia\thanks{The majority of this work was done while the author was supported by Research Pod on Resilience in Brain, Natural, and Algorithmic Systems at the Simons Institute, UC Berkeley.}\\
Carnegie Mellon University\\
{\tt ankitp@cmu.edu} \\
}

\begin{document}

\maketitle

\begin{abstract}
We study the task of noiseless linear regression under Gaussian covariates in the presence of additive oblivious contamination. Specifically, we are given i.i.d.\ samples 
from a distribution $(x, y)$ on $\R^d \times \R$  
with $x \sim \cN(0,\bI_d)$ and $y = x^\top \beta + z$, 
where $z$ is drawn independently of $x$ from an unknown distribution $E$. 
Moreover, $z$  satisfies $\P_E[z = 0] = \alpha>0$.
The goal is to accurately recover the regressor 
$\beta$ to small $\ell_2$-error. 
Ignoring computational considerations, this problem 
is known to be solvable using $O(d/\alpha)$ samples. 
On the other hand, the best known polynomial-time algorithms 
require $\Omega(d/\alpha^2)$ samples. Here we provide formal 
evidence that the quadratic dependence in $1/\alpha$ is 
inherent for efficient algorithms. Specifically, we show 
that any efficient Statistical Query algorithm
for this task requires  
VSTAT complexity
at least $\tilde{\Omega}(d^{1/2}/\alpha^2)$.
\end{abstract}

\setcounter{footnote}{0}
\thispagestyle{empty}

 \newpage

\section{Introduction}
\label{sec:introduction}

\input{intro}

\input{overview-techniques}

\input{related-work}

\section{Preliminaries}

\input{prelims}

\section{Proof of Main Result: \Cref{thm:sq-hardness-discrete-gaussian}}
\label{sec:generalized_SQ_lower_bound}
\input{main-proof}

\input{conc}

\printbibliography

\appendix
\newpage
\input{appendix}

\end{document}

%% file: macros.tex
\usepackage[T1]{fontenc}    %
\usepackage{url}      %
\usepackage{booktabs}     %
\usepackage{nicefrac}     %

\usepackage{amsthm,amsfonts,amsmath,amssymb,epsfig,color,float,graphicx,verbatim,
enumitem}

\usepackage{algpseudocode,algorithm,algorithmicx}  

\usepackage{bbm}
\usepackage{caption}
\usepackage{duckuments}
\usepackage[toc,page]{appendix}
\usepackage{mdframed}
\usepackage{multicol}

\usepackage[
backend=biber,
style=alphabetic,
maxbibnames=15,
maxalphanames=10,
minalphanames=6,
doi=false,
isbn=false,
url=false,
eprint=false,
backref=false,
]{biblatex}
\usepackage{color}
\colorlet{linkcolor}{magenta}
\hypersetup{
    colorlinks=true,
    citecolor=blue,
    linkcolor=magenta,
}

\usepackage[nameinlink,capitalise]{cleveref}

\usepackage{thm-restate}

\newlist{itemizec}{itemize}{2}
\setlist[itemizec,1]{label=\faCaretRight ,wide, parsep= 0.05pt, left = 15pt}

\makeatletter
\newcommand{\AND}{\end{tabular}\par\begin{tabular}[t]{c}}
\makeatother

\def\E{\mathbb E}
\def\P{\mathbb P}
\def\R{\mathbb R}
\def\I{\mathbb I}

\def\N{\mathbb N}
\def\Z{\mathbb Z}
\newcommand{\eps}{\epsilon}

\newcommand\numberthis{\addtocounter{equation}{1}\tag{\theequation}}

\newcommand{\1}{\mathbbm{1}}

\let\vec\mathbf  %

\newcommand{\bA}{\vec{A}} \newcommand{\bB}{\vec{B}} \newcommand{\bC}{\vec{C}}
  \newcommand{\bF}{\vec{F}}
 \newcommand{\bH}{\vec{H}} \newcommand{\bI}{\vec{I}}
\newcommand{\bJ}{\vec{J}}

 \newcommand{\bT}{\vec{T}} \newcommand{\bU}{\vec{U}}
\newcommand{\bV}{\vec{V}}

\newcommand{\bv}{\vec{v}} \newcommand{\bw}{\vec{w}}

\newcommand{\cA}{\mathcal{A}} \newcommand{\cB}{\mathcal{B}} \newcommand{\cC}{\mathcal{C}}
 \newcommand{\cE}{\mathcal{E}} 
  
  \newcommand{\cL}{\mathcal{L}}
 \newcommand{\cN}{\mathcal{N}} 
  
\newcommand{\cS}{\mathcal{S}}  \newcommand{\cU}{\mathcal{U}}
  \newcommand{\cX}{\mathcal{X}}
 \newcommand{\cZ}{\mathcal{Z}}

\newcommand{\poly}{\mathrm{poly}}
\newcommand{\polylog}{\mathrm{polylog}}

\newcommand{\dtv}{\mathrm{TV}}

\def\argmax{\qopname\relax n{argmax}}
\def\argmin{\qopname\relax n{argmin}}

\newcommand{\SDA}{{\mathrm{SDA}}}

\newcommand{\VSTAT}{{\mbox{VSTAT}}}

\newcommand{\Ber}{\mathrm{Ber}}

\newcommand{\littlesum}{\mathop{\textstyle \sum}}

\crefformat{equation}{(#2#1#3)}
\crefname{equation}{Equation}{Equations}
\crefname{lemma}{Lemma}{Lemmata}
\crefname{claim}{Claim}{Claims}
\crefname{fact}{Fact}{Facts}
\crefname{theorem}{Theorem}{Theorems}
\crefname{proposition}{Proposition}{Propositions}
\crefname{corollary}{Corollary}{Corollaries}
\crefname{remark}{Remark}{Remarks}
\crefname{definition}{Definition}{Definitions}
\crefname{question}{Question}{Questions}
\crefname{condition}{Condition}{Conditions}
\crefname{figure}{Figure}{Figures}
\crefname{testingproblem}{Testing Problem}{Testing Problems}
\crefname{subproposition}{Proposition}{Propositions}

\newtheorem{theorem}{Theorem}[section]
\newtheorem{lemma}[theorem]{Lemma}

\newtheorem{claim}[theorem]{Claim}
\newtheorem{proposition}[theorem]{Proposition}

\theoremstyle{definition}
\newtheorem{fact}[theorem]{Fact}
\newtheorem{definition}[theorem]{Definition}
\newtheorem{testingproblem}[theorem]{Testing Problem}

\newtheorem{question}[theorem]{Question}

\newtheorem{remark}[theorem]{Remark}

\allowdisplaybreaks

\theoremstyle{definition}
\newcommand{\sign}{\text{sgn}}

\newenvironment{proofsketch}[1][Proof Sketch]{%
  \renewcommand{\proofname}{#1}%
  \begin{proof}%
}{%
  \end{proof}%
}

\definecolor{Red}{rgb}{1,0,0}
\definecolor{Blue}{rgb}{0,0,1}
\definecolor{DGreen}{rgb}{0,0.55,0}
\definecolor{Purple}{rgb}{.75,0,.25}
\definecolor{Grey}{rgb}{.5,.5,.5}

\def\red{\color{Red}}

\def\grey{\color{Grey}}

\newlist{initemize}{itemize*}{1}
\setlist[initemize]{%
  label=\textbullet,            %
  labelsep=0.5em,               %
  itemjoin=\hspace{1em}         %
}

\newcommand{\distCont}{T}
\newcommand{\distDisc}{Q}
\newcommand{\distDiscT}{\widetilde{\distDisc}}
\newcommand{\hidDisc}{A}
\newcommand{\hidDiscT}{\widetilde{A}}
\newcommand{\hidCont}{B}
\newcommand{\distNull}{P}

\newcommand{\nullSQ}{\mathcal{P}}
\newcommand{\altSQ}{\mathcal{Q}}
\newcommand{\hiddenSQ}{\mathcal{H}}

\newcommand{\DG}[1]{%
 \hypersetup{linkcolor=black}%
  \hyperref[def:discrete-gaussian]{%
    \textsf{\textsc{dg}}'\!\bigl[#1\bigr]%
  }%
}

\newcommand{\NDG}[1]{%
 \hypersetup{linkcolor=black}%
  \hyperref[def:discrete-gaussian]{%
    \textsf{\textsc{dg}}\!\,\bigl[#1\bigr]%
  }%
}

\newcommand{\NDGs}[1]{%
 \hypersetup{linkcolor=black}%
  \hyperref[def:discrete-gaussian]{%
    \textsf{\textsc{dg}}\!\,[#1]%
  }%
}

%% file: intro.tex
Linear regression is a prototypical 
supervised learning task
with a wide range of applications~\cite{Rousseeuw:1987, Dielman01, McD09}.
In the vanilla setting, we are given 
labeled samples $(x^{(i)}, y^{(i)})$,  
where the covariates $x^{(i)}$ are drawn i.i.d.\  
from a distribution on $\R^d$ 
and the labels $y^{(i)}$ are (potentially noisy) 
evaluations of a linear function. 
The goal of the learner is to approximately recover 
the hidden regression vector. 
In this standard setting, linear regression is well-understood both 
statistically and computationally. 
Specifically, under Gaussian covariates {with additive Gaussian noise}, the least-squares estimator 
is computationally efficient and statistically optimal.

In many real-world scenarios, the input data is subject 
to some form of contamination, e.g.,  
errors due to skewed and corrupted measurements, 
making even simple statistical estimation tasks algorithmically challenging.
In the context of linear regression, 
classical computationally efficient estimators 
inherently fail in the presence of data contamination. 
An important goal {in this context} is to understand the  
possibilities and limitations of computationally efficient estimation 
in the presence of contaminated data.

In this work, we study the fundamental problem 
of linear regression with Gaussian covariates in the presence of  
{\em oblivious} {additive} 
contamination in the responses   
(see \Cref{def:estimation-problem}).  
In the oblivious contamination model, an adversary is allowed 
to corrupt a $(1-\alpha)$-fraction of the labels 
{(by adding an adversarially selected 
value to the label)}, 
for some parameter 
$\alpha>0$, and is limited in their capability by requiring the 
contamination be {\em independent} of the samples.
Interestingly, the oblivious 
model information-theoretically allows 
for consistent estimation even for $\alpha \rightarrow 0$. 
This stands in contrast to the more challenging model of 
adversarial contamination~\cite{Hub64,DiaKan22-book}, 
where non-trivial guarantees are impossible  
if more than half of the labels are corrupted. 

To facilitate the subsequent discussion, we define our learning task below. 

\begin{definition}[Noiseless Linear Regression with Oblivious Contamination in Responses]
\label{def:estimation-problem}
Let $\alpha \in (0,1)$ be the probability of inliers.  
Let $E$ be a univariate distribution with $\P_{Z \sim E}(Z=0)\geq \alpha$. 
For $\beta \in \R^d$, we denote by $P_{\beta, E}$  the 
distribution on labeled examples $(x,y) \in \R^d \times \R$ 
defined as follows:
\begin{align*}
x \sim \cN(0,\bI_d)  \text{ and } y = x^\top \beta + Z, \text{ where $Z \sim E$ is independent of {$x$}. }
\end{align*}
Given i.i.d.\ samples $\{(x_i,y_i)\}_{i=1}^n$ from an unknown $P_{\beta^*,E}$, the goal is to construct an estimate $\widehat{\beta}$ such that $\|\widehat{\beta} - \beta^*\|_2$ is small.  
\end{definition}

The model of \Cref{def:estimation-problem} 
goes back to the work of 
Candes and Tao \cite{candes2005decoding}, who studied it 
(for more general design matrices)  
as a classical example of error correction. 
It is also a standard model in face recognition \cite{wright2008dense}, image inpainting \cite{nguyen2013exact}, privacy-preserving data analysis \cite{dwork2007price}, and model repair \cite{gao2020model}. 
A basic result in this area is that the true $\beta$ can be recovered exactly, as long as the design matrix 
satisfies restricted isometry (therefore, for Gaussian design) and the number of nonzero entries of the noise 
is not too large (detailed below) \cite{candes2005decoding,candes2005error,wright2008dense,nguyen2013exact,gao2020model}.
Interestingly, Candes and Tao \cite{candes2005decoding} noted that the model can also be recast as compressed sensing.

The statistical task of linear regression with Gaussian covariates under oblivious contamination 
has been extensively studied 
over the past decade~\cite{Tsakonas14, JaiTK14,BhatiaJK15,BhatiaJKK17, 
suggala2019adaptive,pesme2020online, dalalyan2019outlier, Steurer21Outliers}. 
The oblivious model has also been explored 
for other natural tasks, including PCA, sparse 
recovery~\cite{pesme2020online, d2021consistent}, 
and estimating a signal with additive oblivious 
contamination~\cite{SoSOblivious21}.
While most prior work has 
focused on Gaussian or subgaussian design 
matrices, a more recent line of investigation 
has developed efficient estimators 
in the distribution-free setting under mild 
assumptions~\cite{DiakonikolasKPT23, DiakonikolasKPT23b}.

Let us return to \Cref{def:estimation-problem} and discuss the precise quantitative aspects.
Ignoring computational constraints, the 
sample complexity $n$ required to obtain any non-trivial estimate of $\beta^*$ for the problem of \Cref{def:estimation-problem} is $n=d/\alpha$; 
in fact, $n = \Theta(d/\alpha)$ samples suffice to estimate $\beta^*$ \emph{exactly}. 
In contrast, the best known computationally efficient algorithms require sample complexity of 
$n=\Omega(d/\alpha^2)$ samples~\cite{gao2020model, Steurer21Outliers}.
Interestingly, known polynomial-time 
algorithms using $n=O(d/\alpha^2)$ samples 
succeed even for the (more challenging) {\em noisy} version of the estimation task---where (in addition to oblivious contamination) the clean labels are perturbed by random observation noise (e.g., Gaussian noise).\footnote{To be precise, in the \emph{noisy} version of the problem, the labels are of the form $y = x^\top \beta + \xi+Z$, where $\xi \sim N(0,\sigma^2)$. 
\Cref{def:estimation-problem} corresponds to the important special case of $\sigma=0$.
For the noisy case of $\sigma>0$, the information-theoretic error rate is 
$\|\widehat{\beta}-\beta^*\|_2 = \Theta\big(\sigma \cdot \sqrt{\frac{d}{n \alpha^2}}\big)$.}

While \emph{noisy} linear regression with oblivious contamination 
information-theoretically requires $\Omega(d/\alpha^2)$ samples, this is not the case for the \emph{noiseless} version considered in this work---where, as mentioned above, $O(d/\alpha)$ samples suffice. 
This quadratic gap in $1/\alpha$ 
between the information-theoretic optimum and the sample 
complexity of known polynomial-time algorithms can be  
significant in applications where the fraction of inliers 
$\alpha$ is small. Beyond practical considerations, given 
the fundamental nature of this estimation problem, it is 
natural to ask whether a computationally efficient algorithm 
with (near-)optimal sample complexity (i.e., within 
logarithmic factors of the optimal) exists. 
This leads to the central question motivating our work:
\begin{question}
\label{ques:main-intro}
Does there exist a constant $c>0$ so that for all $d \in \N,\alpha\in(0,1)$, there exists an algorithm, using 
$O(\frac{\poly(d)}{\alpha^{2 - c}})$ samples and 
running in $\poly(d,n)$ time, 
that computes an estimate $\widehat{\beta}$ 
such that $\|\widehat{\beta} - \beta^*\|_2$ is small?
\end{question}
Our main result answers this question in the negative 
for efficient Statistical Query (SQ) algorithms---a broad and well-studied
family of algorithms.

\subsection{Main Result}

To establish our negative result, 
we shall show that even the following (easier) testing task is computationally hard for SQ algorithms:

\begin{testingproblem}[Testing Version of Linear Regression with Oblivious Contamination]
\label{def:lin-regr-oblivious}
Let $\rho>0$ be the signal strength and $\alpha \in (0,1)$ be the inlier probability. 
Let $E$ be a (known) univariate distribution 
on $\R$ that assigns at least $\alpha$ probability 
to $0$. 
Let $R^*_{\rho,E}$ be the univariate distribution of $G+z$, where $G \sim \cN(0,\rho^2)$ and $z \sim E$ independently. The algorithm gets sample access to a distribution $(x,y) \sim \Theta$ with the goal of distinguishing:
\begin{itemize}[leftmargin=2em]
\item ``Null'':  $\Theta = P$, where under $P$: $x \sim \cN(0,\bI_d)$ and $y \sim R^*_{\rho,E}$  independently.
\item ``Alternate'':  First a unit vector $v$ is sampled uniformly, and then conditioned on $v$,  $\Theta = Q_v$,  where under $Q_v$: $x \sim \cN(0,\bI_d)$ and $y = \rho v^\top x + z $, where $z \sim E$ is independent of $x$. 
\end{itemize}
We say that an algorithm $\cA$ succeeds if the failure probability of $\cA$ is less than $1/10$ under both the ``null'' and the ``alternate''. 
\end{testingproblem}
Note that, under the null hypothesis, the features $x$ and the responses $y$ are independent of each other; while under the alternate hypothesis, they follow {the distribution $P_{\beta,E}$ of } \Cref{def:estimation-problem} with $\|\beta\|_2  = \rho$.
We show in \Cref{app:estimation_is_harder_than_testing} that a (computationally-efficient) estimation algorithm {for the task of estimating $\beta$} 
with error $\rho/4$ suffices to (computationally-efficiently) solve the testing problem above.
\begin{proposition}[Efficient Reduction of Testing to Estimation; Informal]
\label{prop:estimation-hardness}
If there exists a computationally-efficient algorithm to compute $\widehat{\beta}$ with $\|\widehat{\beta}-\beta^*\|\leq \rho/4$ with high probability, then it can be transformed into a computationally-efficient algorithm for \Cref{def:lin-regr-oblivious}. 
\end{proposition} 

\noindent {\bf Basics on SQ Algorithms.}
Instead of getting sample access, SQ algorithms~\cite{Kearns98, FelGRVX17} interact with the underlying distribution $D$ 
through the following oracle. 

\begin{restatable}[VSTAT Oracle]{definition}{DefVSTATOracle} \label{def:vstat}
Let $D$ be a distribution on $\cX$. A statistical query is a bounded function $f : \cX \to [0, 1]$.
For a ``simulation complexity'' $m \in \N$, a \VSTAT(m) oracle for the distribution $D$ on the input $f$ returns a value $v$ such that  $|v-\E_D[f]|\leq \max\bigl\{1/m, \sqrt{(\E_D[f] (1 - \E_D[f]))/{m}}\bigr\}$.
\end{restatable}
That is, the $\VSTAT(m)$ oracle returns an estimate of $\E_D[f]$ with 
error comparable to the deviation in Bernstein's inequality for high-probability estimates of taking $m$ i.i.d.\ samples from the Bernoulli 
distribution with bias $\E_D[f]$. 
{We thus refer to $m$ as the simulation complexity.}

A \emph{Statistical Query (SQ) algorithm} is an algorithm 
whose objective is to learn some information about an unknown 
distribution $D$ by making adaptive calls 
to the corresponding $\VSTAT$ oracle.
{The complexity of an SQ algorithm is quantified by the total number of queries to the $\VSTAT$ oracle (viewed as a measure of the algorithm's running time) and the maximum simulation complexity of 
any such query (viewed as a measure of the algorithm's sample complexity).}

{In the context of our learning problem 
(\Cref{def:estimation-problem}), it is worth pointing out the 
following.
First, there exists an inefficient SQ algorithm
{with small simulation complexity, 
which in particular can be simulated using $\widetilde{O}(d/\alpha)$ many i.i.d.\ samples}
(see 
\Cref{app:inefficient-sq}). 
{Second, there exist efficient SQ algorithm whose simulation complexity matches the sample complexity $\widetilde{O}(d/\alpha^2)$ of known efficient algorithms}
(see \Cref{app:efficient-sq}).}

{With this context,} our main result is the following: 

\begin{theorem}[SQ Hardness of \Cref{def:lin-regr-oblivious}; informal]
\label{thm:sq-hardness-discrete-gaussian}
Consider the %
\Cref{def:lin-regr-oblivious}.
Suppose that (i) $\alpha \gg \frac{1}{d^{\polylog(d)}}$ (i.e., the fraction of inliers is not too tiny) and 
(ii) $\rho = \widetilde{\Theta}(\alpha)$. 
Then there exists a distribution $E$  
satisfying $\P_{Z \sim E}(Z = 0) \geq \alpha$ such 
that any SQ algorithm that solves 
\Cref{def:lin-regr-oblivious} 
either
\begin{itemize}[leftmargin=2em]
    \item uses $d^{\Omega(\log^2(d/\alpha))}$ many queries, or
    \item uses at least one query to $\VSTAT(m)$ 
for $m = \widetilde{\Omega} (\sqrt{d}/\alpha^2)$.
\end{itemize}
\end{theorem}

{Informally speaking, \Cref{thm:sq-hardness-discrete-gaussian} 
shows that no SQ algorithm can solve the testing problem 
(and, via \Cref{prop:estimation-hardness}, the estimation 
problem of approximating $\beta^{\ast}$) 
with less than super-polynomial in $d$ many queries, 
unless using queries whose simulation complexity is at least
$\widetilde{\Omega} (\sqrt{d}/\alpha^2)$.}
That is, either the algorithm ``uses''
$\widetilde{\Omega} (\sqrt{d}/\alpha^2)$ 
many ``samples'' (in the sense of simulation complexity mentioned above) or it takes super-polynomial ``time'' (in the sense of number of queries). 
{We thus obtain evidence that the quadratic dependence in $1/\alpha$ on 
the sample size is required for computationally efficient algorithms. 

It is worth noting that the SQ-hard instances that we construct for 
the testing problem {\em are} efficiently solvable with 
$\widetilde{O}(\sqrt{d}/\alpha^2)$ samples. We conjecture that the correct 
dependence on $d$ is in fact linear (i.e., 
an $\Omega(d/\alpha^2)$ lower bound on the computational sample 
complexity). This is left as an interesting question for future work (see \Cref{sec:conc}). 
}

Finally, while the focus of this work is on the SQ model, 
SQ-hardness results typically translate to quantitatively 
similar hardness for low-degree polynomial 
tests~\cite{Hopkins-thesis,KunWB19}, via the 
work of~\cite{BreBHLS21}. While we do not 
establish a formal theorem in this regard, we 
believe that our SQ-hard instances are also 
hard for low-degree polynomials.

%% file: overview-techniques.tex
\subsection{Overview of Techniques}
\label{sub:overview_of_techniques}

We wish to show that it is hard to solve \Cref{def:lin-regr-oblivious} 
with fewer than $\sqrt{d}/\rho^2$ samples 
(we will ultimately set $\rho = \widetilde{\Theta}(\alpha)$).
The first question we face is to make a judicious choice 
of the {contamination} distribution $E$ that 
\hypertarget{Crit1}{{\color{linkcolor}(I)}} satisfies our noise model, 
namely $\P_{Z \sim E}(Z=0)\geq \alpha$; 
and \hypertarget{Crit2}{{\color{linkcolor}(II)}} it is SQ-hard 
to distinguish the null and alternate hypotheses.

\paragraph{Choice of {Contamination} Distribution: Intuition.}
{A natural first step to consider} 
is what happens if {we select the contamination distribution $E$} 
to be the standard Gaussian, 
i.e., $E = \cN(0,1)$. In this case, the testing task corresponding 
to \Cref{def:lin-regr-oblivious}
is {\em information-theoretically} 
{impossible} with $o(\sqrt{d}/\rho^2)$ samples.
Unfortunately, this {choice} does not fit our criterion \hyperlink{Crit1}{(I)}, requiring that the {contamination distribution} 
must be \emph{exactly} $0$ with probability at least $\alpha$.

Inspired from the information-theoretic sample complexity lower bound for 
the Gaussian {contamination} setting, we instead consider 
a scenario where the {contamination} is given by a distribution $E$, 
which is a {\em discrete} Gaussian with spacing $s$ 
(see \Cref{def:discrete-gaussian}). Heuristically, the
discrete Gaussian
approximately matches its low-degree moments with the {continuous} Gaussian case, and thus, 
{prior work~\cite{DiaKS17} hints}
that it is SQ-hard to distinguish between the cases 
of discrete Gaussian and continuous Gaussian contamination. Since the case of 
continuous Gaussian contamination information-theoretically requires
$\Omega(\sqrt{d}/\rho^2)$ samples, {intuitively we are moving in the right direction.} Note that the aforementioned discrete Gaussian $E$ 
assigns probability $\Omega(s)$ to $0$.
Taking $s = \Theta(\alpha)$, we simultaneously satisfy  criterion \hyperlink{Crit1}{(I)} above and have a reasonable chance of satisfying \hyperlink{Crit2}{(II)}. 

The above is the key intuitive idea underlying our proof. 
However, there are a number of important technical steps 
required to make the analysis work towards 
satisfying \hyperlink{Crit2}{(II)}.

\paragraph{Discrete Noise and Non-Gaussian Component Analysis.}
For a unit vector $v$, let $\distDisc_v$ be the distribution over $(x,y)$ 
such that $y = \rho v^\top  x + Z$, where $Z\sim E$ independently of $x$ 
and $E$ is the suitable discrete Gaussian distribution. 
Let $\distNull$ be the distribution over $(x,y)$ corresponding to the null hypothesis, 
namely  $x$ and $y$ are independent with correct marginals (i.e., $x~\sim \cN(0,\bI_d)$ and $y\sim \rho^2 G + Z$, where $G\sim\N(0,1)$ and $Z\sim E$ are independent).
We wish to show that it is SQ-hard to distinguish between $\distDisc_v$, 
for random $v$, and $\distNull$.
We note that conditioning on the value of $y$, 
$\distDisc_v$ is a standard Gaussian in the directions orthogonal to $v$ 
and is given by some known distribution, $\hidDisc_y$, in the $v$–direction. 
This means that the testing problem we are considering is effectively 
a {\em conditional} Non-Gaussian Component Analysis (NGCA) 
problem 
(\Cref{def:ngca-lin-regr}). Unfortunately, {there are several technical obstacles} preventing us from 
applying existing tools in the literature~\cite{DiaKS17,DiaKS19}.

{The first technical hurdle arises from the fact}  
that $\hidDisc_y$ is a discrete distribution, and in particular has infinite chi-squared norm with respect to the standard Gaussian. 
In particular, this means that the standard SQ–dimension 
related techniques for proving lower bounds will not work here. 
Instead, we need to {leverage and} adapt the recent work of \cite{DiaKRS23} that {directly uses Gaussian Fourier analysis}  
to establish SQ lower bounds even when the chi-squared distance 
is infinite. Unfortunately, the latter work \cite{DiaKRS23} 
does not give SQ-lower bounds for {\em conditional} 
Non-Gaussian Component Analysis tasks 
(as the one we are dealing with here).  Consequently, 
we will require a careful adaptation of their techniques in our context.

\paragraph{Connection with Continuous Gaussian Contamination.}
A key requirement for the Gaussian Fourier analysis to go through 
in \cite{DiaKRS23} is that $\hidDisc_y$'s have well-behaved moments.
Unfortunately, an additional technical challenge arising in our context is that 
directly bounding the relevant moments of $\hidDisc_y$ (which belongs to the family of discrete Gaussians) is challenging. 

Instead, for the purpose of the analysis, we again leverage the connection with continuous Gaussian contamination. 
Specifically, we choose $\hidCont_y$ to be a continuous Gaussian counterpart 
of the discrete Gaussian $\hidDisc_y$.
Let the resulting distribution on $(x,y)$ be $\distCont_v$ 
(which is a continuous counterpart of $\distDisc_v$). 
Note that this is again an instance of conditional NGCA. 
Since the $\hidCont_y$'s are now (continuous) Gaussians 
(and hence satisfy many desirable properties, e.g., continuity), 
it can be shown that if $v$ and $w$ are nearly orthogonal vectors, 
$\distCont_v$ and $\distCont_w$ will have small chi-squared inner product 
with respect to $P$. 

\paragraph {Hardness of Continuous Noise Contamination.} The fact that two random unit vectors have small inner product with high probability can be used to show 
that the task of testing between $P$ and $\{\distCont_v\}_{v \sim \cS^{d-1}}$ 
has large SQ dimension. 
This implies SQ-hardness of this basic testing problem. 
In fact, it will imply the more powerful result that for any bounded function $f$, with high probability over $v$, the expectations $\E_{\distCont_v}[f]$ and $\E_P[f]$ cannot be distinguished by a $\VSTAT(o(m_0))$ query for $m_0:=\sqrt{d}/(\rho^2\cdot\log^4d)$; see \Cref{lem:sq-hardness-continuous-noise}.

 \paragraph{Quantitative Relationship between Discrete and Continuous Gaussian Noise.} 
We now return to the challenge of computing moments of discrete Gaussians 
$\hidDisc_y$ (for performing Gaussian Fourier analysis). We resolve this issue 
by comparing these moments to the moments of $\hidCont_y$. 
As $\hidDisc_y$ will be a discrete version of the Gaussian $\hidCont_y$, 
this relationship will be relatively manageable to prove. 
We then combine this ingredient with techniques involving  Hermite analysis 
from \cite{DiaKRS23} to show the following: 
for any bounded test function $f$, 
with high probability over the choice of a random $v$, 
it holds that $|\mathbf{E}_{\distDisc_v}[f] - \mathbf{E}_{\distCont_v}[f]|$ is tiny (inverse super-polynomial in $m_0$)  as long as $s \ll \tfrac{\rho}{\polylog(d)}$ (\Cref{prop:distinguish-discrete-vs-cont}).

 \paragraph{Putting Everything Together.} 
{Combining the above}, we obtain the following: for any $f$, with high probability over random $v$, it holds that (i)  
$|\mathbf{E}_{\distDisc_v}[f] - \mathbf{E}_{\distCont_v}[f]|$ 
is {inverse} super-polynomially small in $m_0$, and (ii) 
$|\mathbf{E}_{\distCont_v}[f] - \mathbf{E}_{P}[f]|$ is smaller than the threshold for $\VSTAT(o(m_0))$. Therefore, by a union bound and a triangle inequality, 
it follows that with high probability 
$|\mathbf{E}_{\distDisc_v}[f] - \mathbf{E}_P[f]|$ is also  
smaller than the threshold for $\VSTAT(o(m_0))$, implying SQ-hardness (\Cref{prop:sq-lower-bound-intro}).

%% file: related-work.tex
\subsection{Related Work}

Our work is broadly situated in the field of robust statistics, which has a long history dating back to Huber and Tukey~\cite{Hub64,Tuk60}.
Robust statistics aims to design estimators that are tolerant to data contamination.
Focusing on high-dimensional data, our work studies the statistical and computational aspects of robust estimation, which has seen a flurry of work in the last decade since \cite{DiaKKLMS16-focs,LaiRV16}; see \cite{DiaKan22-book} for a recent book on this topic.
For designing robust estimators, the choice of contamination model naturally plays a crucial role. 
This work is part of a broader effort to understand computational and statistical aspects of natural, not fully adversarial, contamination models; see, e.g.,~\cite{BhatiaJK15,BhatiaJKK17,ZhuJS19,DGT19,DiakonikolasKane22-massart-sq,DKMR22, DKRS22, DKKTZ22, DiakonikolasKPT23, DDKW23, DDKW23b, ma2024estimationmissingcompletelyrandom,NieGS24,PenPit24,DZ24, KotGao25,DiaIKP25-mean-shift}.

Historically, the prototypical contamination model in robust statistics has been Huber's contamination model~\cite{Hub64}, which was strengthened to total variation distance~\cite{Huber65} and strong contamination models~\cite{DiaKKLMS16-focs}.
The task of linear regression under these contamination models is now well understood both statistically~\cite{CheGR16} and computationally~\cite{DiaKS19,PenJL20,DiaKPP23-huber-optimal}.
As mentioned earlier, it is information-theoretically impossible to achieve consistency in these models 
if the proportion of contamination  
is bounded away from zero.
Thus, an important direction is to understand the possibilities and limitations in other, less adversarial, contamination models.
The \emph{oblivious} adversary studied here 
is one such model, and indeed it does lead to consistent estimation even when the oblivious outliers constitute the majority of the observed data; see the discussion below \Cref{def:estimation-problem}.
Our work shows that while this weaker contamination model is benign from the perspective of information-theoretic rates, it does present surprising information-computation tradeoffs.

%% file: prelims.tex
\paragraph{Notation.}  
For a univariate distribution $E$, we define $R^*_{\rho,E}$ to be the univariate distribution of $G+z$, where $G \sim \cN(0,\rho^2)$ and $z \sim E$ independently. 
For two vectors $v$ and $w$ in $\R^d$, we use $\langle v, w\rangle$ and $u^\top w$ interchangeably to denote the standard inner product $\sum_{i \in [d]} v_iw_i$. 
A degree-$k$ tensor, or $k$-tensor in short, in $d$-dimensions  $\bv$ is an element of $(\R^{d})^{\otimes k}$ 
with entries $(\bv_{i_1,\dots,i_k})_{i_1 \in [d],\dots, i_k \in [d]}$.
For a vector $v$, we use $v^{\otimes k}$ to denote the $k$-tensor with entries $\prod_{\ell=1}^d v_{i_\ell}$.
For two $k$-tensors $\bv$ and $\bw$, we use $\langle \bv,\bw\rangle$ to denote the inner product $\sum_{i_1,\dots, i_k}\bv_{i_1,\dots,i_k}\bw_{i_1,\dots,i_k}$ and use $\|\bv\|_2:= \sqrt{\langle \bv,\bv\rangle}$.
A $k$-tensor function $\bF:\cX \to (\R^d)^{\otimes k}$ maps each $x \in \cX$ to a $k$-tensor.

We use $\widetilde{\Omega}$, $\widetilde{\Theta}$ notation to hide $\polylog$ factors in the arguments. 
For two non-negative functions,  $a$ and $b$, we use $a \lesssim b$ (similarly $a \gtrsim b$) to say that there exists a constant (independent of other problem parameters) such that $a \leq C b$ (respectively, $a\geq Cb)$; if $a\lesssim b$ and $b\gtrsim a$, then we say $a \asymp b$.

\paragraph{Concentration.} 
We say that a real-valued random variable $X$ is $\sigma$-subgaussian if $\P(|X]| >  t ) \lesssim \exp(-t^2/\sigma^2)$ for all $t> 0$.
We shall  use the following fact about the moment generating function of $X^2$ for subgaussian random variables: 
\begin{restatable}{fact}{FactSubgaussian}
    \label{fact:subgaussian-square}
    There exists a finite constant $a_0 > 0$ such that if $X$ is $\sigma$-subgaussian then $\left|\E\Big[e^{a \frac{X^2}{\sigma^2}}\Big] -1 \right| \lesssim |a|$ for $|a| \leq a_0$.
\end{restatable}
The next fact follows from H\"older's inequality:
\begin{restatable}{fact}{FactSubgaussianConc}
    \label{fact:subgaussian-cond}
    Let $X$ be a $\sigma$-subgaussian random variable. Then, for any event $\cE$, we have that $\E[|X| \big\vert \cE] \lesssim \sigma \sqrt{\log(1/\min(\P(X \in \cE), \P(X \not\in \cE))} $.
\end{restatable}

 \paragraph{Hermite Polynomials.} 
For $k \in \N$, we use $h_k: \R \to \R$ to denote the $k$-th normalized probabilist's Hermite polynomial, which is the degree-$k$ polynomial defined by $h_k(x) := \frac{1}{\sqrt{k!}} (-1)^k e^{x^2/2} \frac{d^k}{dx^k} e^{-x^2/2}$. 
We shall also use  the $k$-th Hermite tensor function $\bH_k(x)$ as defined in \cite[Definition 2.2]{DiaKRS23}.

 \paragraph{Fourier Analysis.} For a distribution $P$ on a domain $\cX$, we use $L^2(\cX, P)$ to denote the space of all functions $f:\cX \to \R$ with $\E_{x \sim P}[f^2(x)]< \infty$.
For two functions $f,g \in L^2(\cX, P)$, we use $\langle f, g\rangle_P $ to denote the inner product $\E_{x \sim P}[f(x)g(x)]$ and $\|f\|_{L_2(P)}$ to denote $\langle f, f \rangle _P$.
For a function $f:\R^d \to \R$ and an $\ell \in \N$, we define $f^{\leq \ell}$ to be the degree-$\ell$ Hermite approximation function  $f^{\leq \ell}(x):= \sum_{k=0}^\ell \langle \bA_k, \bH_k(x)\rangle$ where $\bA_k := \E_{x \sim P}[f(x) \bH_k(x)]$ is a $k$-tensor, which is equal to $\langle f, \bH_k\rangle_P$ elementwise. 
We extend this definition to $f:\R^{d} \times \R$ as follows: First, for each $y \in \R$, we define $f_y:\R^d \to \R$ as $x \mapsto f(x,y)$ and then define $f^{\leq \ell}(x,y) := f_y(x)^{\leq \ell}$; that is, for each $y$, we perform degree-$\ell$ approximation of $f_y$.  
We use $f^{>\ell} := f - f^{\leq \ell}$ to denote the residual.
\begin{restatable}{fact}{FactHermite}
\label{fact:Hermite}
 For every function $f:\R^d \to [-1,1]$,  $\|f^{>\ell}\|_{L_2(\cN(0,\bI_d))} \to 0$ as $\ell \to \infty$.
Furthermore, for any $f:\R^d \times \R \to [-1,1]$ and univariate measure $R$, $\|f^{>\ell}\|_{L_2(\cN(0,\bI_d) \times R)} \to 0$ as $\ell \to \infty$.
\end{restatable}

\vspace{-0.2cm}

\subsection{Statistical Query Algorithms}

\label{sub:statistical_query_algorithms}

Recall that, instead of getting direct sample access, SQ algorithms interact with the underlying distribution through a $\VSTAT$ oracle (see Definition~\ref{def:vstat}).
Observe that there are many ways of implementing a $\VSTAT(m)$ oracle, especially when the SQ algorithm $\cA$ makes multiple queries---all we require is that each response is a valid $\VSTAT(m)$ response to each query.
Furthermore, the algorithm should successfully solves the inference task as long as each response to its query is a valid $\VSTAT(m)$ response irrespective of how these responses were generated.

We state the preliminaries of SQ for the following generic testing problem.
\begin{testingproblem}[Generic Testing Problem]
\label{def:generic-sq}
Let $\nullSQ$ and $\{\altSQ_v\}_{v \in \cS^{d-1}}$ be distributions over a domain $\cZ$, which correspond to ``null'' and ``alternate'', respectively.
\begin{itemize}[leftmargin=*]
  \item First sample $\Gamma \sim \Ber(1/2)$ and  $v \sim \cS^{d-1}$ independently (unknown to the statistician).

\item   Then set $\Theta = \nullSQ$ (``null'') if $\Gamma = 0$ and  $\Theta = \altSQ_v$ (``alternate'') otherwise. 

\item The statistician gets (either sample/oracle) access to the distribution $\Theta$ and generates $\widehat{\Gamma} \in \{0,1\}$ using an algorithm $\cA$.  
\end{itemize}
\end{testingproblem}

We say an SQ algorithm $\cA$ successfully solves a problem with query complexity $q$ and simulation complexity $m$ if, for any $\VSTAT(m)$ oracle on the underlying distribution ($\Theta$ above),
$\cA$ iteratively (potentially also adaptively and randomly) makes queries $f_1,\dots, f_q$ (each $f_i$ is bounded in $[0,1]$ and could depend on the previous queries and their responses) and outputs $\widehat{\Gamma}$ such that $\P(\widehat{\Gamma} \neq \Gamma) \leq 0.1$.

An SQ lower bound is an information-theoretic lower bound of the following form: any successful SQ algorithm $\cA$ must have either $q \geq q_0$ or $m \geq m_0$.
Observe that to prove an SQ lower bound of this flavor, it suffices to construct one $\VSTAT(m_0)$ oracle 
so that no SQ algorithm can reliably identify $\Gamma$ with at most $q_0$ queries.
In the remainder of this section, we detail the technical results required for proving such lower bounds.

\begin{restatable}[Pairwise Correlation]{definition}{DefPairwiseCor}
For a reference distribution $\nullSQ$, and candidate distributions $\altSQ_1$ and $\altSQ_2$,
the pairwise correlation between $\altSQ_1$ and $\altSQ_2$ with respect to $\nullSQ$ is defined as 
$\chi_\nullSQ(\altSQ_1,\altSQ_2) := \E_{Z \sim \nullSQ}\big[ \frac{q_1(Z) q_2(Z)}{p^2(Z)} -1\big]$,
where $q_1(\cdot), q_2(\cdot), p(\cdot)$ denote the densities of $\altSQ_1,\altSQ_2$, and $\nullSQ$ with respect to a common measure, respectively.
When $\altSQ_1=\altSQ_2$, the pairwise correlation becomes the same as the 
$\chi^2$-divergence between $\altSQ_1$ and $\nullSQ$, i.e., 
$\chi^2(\altSQ_1,\nullSQ) =  \int_{\cZ}\frac{q_1^2(x)}{p(x)} dx - 1$.
\end{restatable}

The statistical dimension is then defined using these pairwise correlations:
\begin{restatable}[Statistical dimension from \cite{BreBHLS21}]{definition}{DefStatDimension}
\label{def:SDA}
The  statistical dimension of \Cref{def:generic-sq} 
 at simulation complexity $m$
 is defined as:    
\vspace{-0.1cm}
\begin{align*}
\textstyle
\SDA(m) := \max\Big\{q \in \N: { \sup_{\cE: \P_{v,v'}(\cE) \geq 1/q^2}} \E_{v,v'}\left[ \left| \chi_\nullSQ( \altSQ_v,\altSQ_{v'}) \right| \big|\cE \right] \leq \frac{1}{m}    \Big\}\,,
\end{align*}
where (i) $v,v'$ are two uniform, independent unit vectors, and (ii) the inner supremum is taken over events $\cE \subset \cS^{d-1} \times \cS^{d-1}$ on $v$ and $v'$ (drawn i.i.d.\ from the unit sphere) that have probability at least $1/q^2$. 
\end{restatable}
We now define the notion of success of a query $f$ that will be useful to us:
\begin{restatable}[Success of a query on a distribution]{definition}{DefNotionSuccess}
\label{def:notion-of-success}
We say that a query $f:\cZ \to [0,1]$ succeeds on distinguishing $\altSQ_v$ and $\nullSQ$ with simulation complexity $m$, denoted by the event $\cE_{f,v,m}$,  if 
$|\E_{\altSQ_v}[f(Z)] - \E_{\nullSQ}[f(Z)]| \geq \max\Big(\frac{1}{m}, \min\Big(\sqrt{\tfrac{a(1-a)}{m}}, \sqrt{\tfrac{b(1-b)}{m}} \Big) \Big)$ for $a:= \E_\nullSQ[f]$ and $b:= \E_{\altSQ_v}[f]$.
\end{restatable}
We are now equipped to state the generic SQ lower bounds that we will use repeatedly.
These results follow from the ideas implicit in prior works \cite{FelGRVX17,BreBHLS21,DiaKRS23} and their proofs are deferred to \Cref{app:statistical_query_algorithms}.
\begin{proposition}[Generic SQ Lower Bound]
\label{prop:sq-lower-bound-intro}
Consider \Cref{def:generic-sq}. Then
\begin{enumerate}[leftmargin=*, label=(C.\Roman*), ref=\theproposition\,(C.\Roman*)]
  \item \label[subproposition]{prop:SDA-upper-bound} For  any query $f: \cZ \to [0,1]$,
 $\P_{v \sim \cS^{d-1}}\left( \cE_{f,v.m}  \right) \leq \frac{1}{\SDA(7m)}$.
\item \label[subproposition]{eq:sq-mid-step-intro}
 Suppose for all queries $f:\cZ \to [0,1]$, it holds that  $\P_{v \sim \cS^{d-1}}\left( \cE_{f,v,m}  \right) \leq \frac{1}{q}$.
Then any  SQ algorithm $\cA$ that solves \Cref{def:generic-sq} must
use either $\Omega(q)$ queries in expectation or at least one query as powerful
as $\VSTAT(m + 1)$.
\end{enumerate}
\end{proposition}

\paragraph{Non-Gaussian Component Analysis.}
We will primarily consider \Cref{def:generic-sq} of a particular form called Non-Gaussian Component Analysis (NGCA).
We begin by defining the following High-Dimensional Hidden Direction Distribution: 
\begin{definition}[High-Dimensional Hidden Direction Distribution] \label{def:high-dim-distribution}
For a unit vector $v \in \R^d$ and a distribution $\hiddenSQ$ on the real line, we define $P^{\hiddenSQ}_v$ to be the distribution over $\R^d$, where $P^{\hiddenSQ}_v$ is the product distribution whose orthogonal projection onto the direction of $v$ is $\hiddenSQ$, 
and onto the subspace perpendicular to $v$ is the standard $(d{-1})$-dimensional normal distribution. 
In particular, if $\hiddenSQ$ is a continuous distribution with probability density function (pdf) $\hiddenSQ(x)$, 
 then $P^{\hiddenSQ}_v(x)$ has the pdf $\hiddenSQ(v^\top x) \phi_{\bot v}(x)$, where $\phi_{\bot v}(x) = \exp\left(-\|x - (v^\top x)v\|_2^2/2\right)/(2\pi)^{(d-1)/2}$.
\end{definition}
\cite{DiaKS17} established SQ lower bounds for the NGCA problem, where the null and the alternate are  $\cN(0,\bI_d)$ and $\{P^\hiddenSQ_v\}_{v \sim \cS^{d-1}}$, respectively and (i) $\hiddenSQ$ (nearly) matches many moments with $\cN(0,1)$ and (ii) has finite $\chi^2(\hiddenSQ,\cN(0,1))$.
For linear regression, we will need the following generalization:  
\begin{testingproblem}[Conditional NGCA]
\label{def:ngca-lin-regr}
Let $\{\hiddenSQ_y\}_{y\in\R}$ be a family of univariate distributions and $R$ be a univariate distribution.
Consider \Cref{def:generic-sq} over $(x,y)$ on the domain  $\R^d \times \R$ with 
\begin{itemize}[leftmargin=*]
  \item (``Null'') Under $\nullSQ$:  $x \sim \cN(0,\bI_d)$ and $y \sim R$ independently.
  \item (``Alternate'') Under $\altSQ_v$: $y \sim R$ and conditioned on $y=y_0$, $X|_{y=y_0} \sim P^{\hiddenSQ_{y_0}}_v$. 
\end{itemize}

\end{testingproblem}
Building on \cite{DiaKS17}, \cite{DiaKS19} showed SQ-hardness for the problem above if (i) $\hiddenSQ_y$ matches its low-degree moments with $\cN(0,1)$ for (nearly) all $y\in \R$, and (ii) $\chi^2(\altSQ_v, \nullSQ)<\infty$. Unfortunately, neither of these conditions holds for us, and we need more flexible and powerful tools to bypass these limitations. 

\subsection{Discrete Gaussians} 
\label{ssub:discrete_gaussian}

The following definition of Discrete Gaussian distributions will be central in our analysis. 
\begin{definition}[Discrete Gaussian Distributions]
\label{def:discrete-gaussian}
For a center $\mu \in \R$, deviation $\sigma > 0$, base $\theta$, and spacing $s > 0$,
define the positive measure $\DG{\mu, \sigma, \theta, s}$ to be the positive measure over $\theta + s \Z$ that assigns mass $s \phi_{\mu, \sigma}( \theta + si )$ for all $i \in \Z$; here, $\phi_{\mu, \sigma}$ denotes the pdf of the Gaussian distribution with mean $\mu$ and standard deviation $\sigma$.
We use $\NDG{\mu,\sigma,\theta,s}$ to denote the normalized probability distribution and refer to it as a Discrete Gaussian.
\end{definition}
The next facts state that the Discrete Gaussian behaves similarly to the vanilla (continuous) Gaussians with respect to translation and low-degree polynomials: 
\begin{restatable}[Translation of Discrete Gaussian]{fact}{FactDiscreteGaussian}
\label{fact:discrete-translation}
For any $\theta \in \R$, $s > 0$, $\mu \in \R$, $\sigma \in \R_+$, the random variables $X \sim \NDG{\mu, \sigma, \theta, s}$ and $X':= \sigma Y + \mu$ for $Y \sim \NDG{0, 1, \theta', s'}$ with $\theta'= (\theta - \mu)/\sigma$ and $s' = s/\sigma$
have the same law.
\end{restatable}
\begin{fact}[{\cite[Fact C.3]{DiaKRS23} and \cite[Lemma 3.12]{DiakonikolasKane22-massart-sq}}] We have the following:
\label{fact:polynomial-discrete-vs-gaussian}
\label{fact:polynomial-discrete-vs-gaussian-appendix}
\begin{itemize}[leftmargin=2em]
    \item For any polynomial $p$ of degree at most $k$, any $\theta \in \R$ and $s > 0$, we have that
$$\left|\E_{G \sim \cN(0,1)} [p(G)] - \E_{Y \sim \NDG{0, 1, \theta, s}} [p(Y)]\right| \lesssim \sqrt{\E_{G \sim \cN(0,1)}[p^2(G)]} k! 2^{O(k)} \exp(- \Omega(1/s^2)).$$
\item (Monomials for the unnormalized measure) For any $k \in \N \cup \{0\}$ and $s \geq 0$:
$$\left|\E_{G \sim \cN(0,1)} [G^k] - \E_{Y \sim \DG{0, 1, \theta, s}} [Y^k]\right| \lesssim  k! (O(s))^k \exp(- \Omega(1/s^2)),$$
where $\E_{Y \sim \DG{0, 1, \theta, s}}[f(Y)$ refers to the Lebesgue integral of $f$ with respect to the positive measure $\DG{0, 1, \theta, s}$.
In particular, the total mass of $\DG{0,1,\theta, s}$ is $1 \pm \exp(-\Omega(1/s^2))$.

\end{itemize}
\end{fact}

%% file: main-proof.tex
In this section, we will prove our main result \Cref{thm:sq-hardness-discrete-gaussian}.
The first step is to make a judicious choice of the noise distribution $E$.
For reasons outlined in \Cref{sub:overview_of_techniques}, we choose $E$ to be a discrete Gaussian with $\sigma^2\approx 1$ and spacing $s$ (eventually set to $\widetilde{\Theta}(\alpha)$).

As the second step, we note that the resulting testing problem is an instance of conditional NGCA.
\begin{testingproblem}[NGCA with Discrete Gaussian]
\label{def:ngca-lin-regr-discrete}
For $y \in \R$, define the distribution $A_y:= \NDG{\mu_y,  \sigma,\theta_y,s'}$ with parameter values in \Cref{parameters}.
Consider \Cref{def:ngca-lin-regr} with
    \begin{itemize}
        \item (Marginal of $y$)  $R := R^*_{\rho,E}$ with $E = \NDG{0,\sigma,0,s}$.\footnote{Recall that $R^*_{\rho,E}$ is defined to be the univariate distribution of $x + z$ for $x \sim \cN(0,\rho^2)$ and $z \sim E$ independent of each other.}
        \item (Conditional NGCA) For each $y\in\R$,  $\hiddenSQ_y$ is equal to
   $A_y=\NDG{\mu_y,  \sigma,\theta_y,s'}$. 
    \end{itemize}
    We denote the corresponding null by $\distNull$ and the alternate for direction $v$ by $\distDisc_v$.
\end{testingproblem}
We mention the parameter choices below:
\begin{mdframed}
\begin{definition}
\label{parameters}
Let signal strength $\rho \in (0,\rho_0)$ for sufficiently small $\rho_0>0$, standard deviation $\sigma\in (0.5,1)$, spacing $s \in (0,1)$ satisfy the following values:

\noindent
\begin{minipage}[t]{0.48\linewidth}
  \begin{itemize}[nosep]
    \item $\sigma = \sqrt{1 - \rho^2}$,
    \item $s' = s/\rho \leq 0.001$,
  \end{itemize}
\end{minipage}\hfill
\begin{minipage}[t]{0.48\linewidth}
  \begin{itemize}[nosep]
    \item $\mu_y := \rho\,y$,
    \item $\theta_y = y/\rho$.
  \end{itemize}
\end{minipage}\end{definition}
\end{mdframed}
That is, for each $y$, the conditional distribution of the covariates in the hidden direction is a discrete Gaussian with mean $\mu_y$ (scaling linearly with $y$) and standard deviation $\sigma$ (slightly smaller than $1$).
While these parameters might look a bit obscure, they perfectly resemble the typical setting of $E= \cN(0,\sigma^2)$.\footnote{The conditional distribution of $x|y$ in the hidden direction $v$ for $E = \cN(0,\sigma^2)$ would be $\cN(\mu_y, \sigma^2)$~\cite[Fact 3.3]{DiaKPPS21}.}

The next result  shows that \Cref{def:lin-regr-oblivious} is equivalent to \Cref{def:ngca-lin-regr-discrete}.

\begin{restatable}{proposition}{LemAyForLinRegression}
\label{lem:A_y-for-lin-regression}
 \Cref{def:ngca-lin-regr-discrete} is equivalent to \Cref{def:lin-regr-oblivious} when $E= \NDG{0,\sigma,0,s}$. 
\end{restatable}
\begin{proof}
First, by definition the distributions $P$ under \Cref{def:lin-regr-oblivious} and \Cref{def:ngca-lin-regr-discrete} are the same.
For $Q$, we shall do the calculations explicitly.

As a starting point, it is easy to see that the conditional distribution of $X$ given 
$y$ under \Cref{def:lin-regr-oblivious} is an instance of NGCA, as in \Cref{def:ngca-lin-regr-discrete}.
To see this, define $x' = v^\top x$ to be the projection of $x$ along $v$, and define $x_{\perp} = x - (v^\top x) v$ to be its orthogonal projection.
Observe that $x'$ and $x_\perp$ are distributed as standard (multivariate) Gaussian and are independent of each other (because $X \sim \cN(0,\bI_d)$).
Hence, the conditional distribution of $y$ given $X \equiv (x',x_\perp)$ can be written as $y = \rho x' + Z$, implying that $y$ is independent of $x_\perp$.
Therefore, the conditional distribution of $X$ given 
$y = y_0$ follows like a standard (multivariate) Gaussian in subspace orthogonal to $v$.
Along the direction $v$, the distribution of $X$ is equivalent to the conditional distribution of $x'$ given $y$, which we denote by $\widetilde{J}_y$.
Our goal is to show that $\widetilde{J}_y$ is equal to $\NDG{\mu_y, \sigma , \theta_y, s'}$ as in \Cref{def:ngca-lin-regr-discrete}.

Observe that marginal distribution of $Y = \cN(0,\rho^2) + \NDG{0,\sigma,0,s}$ (independently drawn) is a Gaussian mixture with countable components: 
\begin{align*}
Y \sim \sum_{i \in \Z} w(i) \cN(si, \rho^2),
\end{align*}
with $w(i) = c s (2 \pi)^{-1/2}\exp(-s^2i^2/(2\sigma^2))$,
where $c$ denotes the normalization constant.
Since $Z$ is discrete over the domain $s \Z$,
the conditional distribution of $X$ given $Y= y_0$ is discrete with support $(y_0 - s \Z)/\rho = \theta_{y_0} - s'\Z$, which is the same support as $\NDG{\mu_{y_0}, \sigma , \theta_{y_0}, s'}$.
For any $x_0$ in this discrete set,
the conditional probability of $X = x_0$ given $y = y_0$ is given by the following (where we hide multiplicative terms that do not depend on $x_0$ under the normalization constant):
\begin{align*}
 \P(X = x_0|X + Z = y_0) &\propto f_X(x_0) \P(Z = y_0 - \rho x_0)\\
&\propto  \left(\exp(-x_0^2/2)\right)  \left( w(y_0 - \rho x_0) \right)\\
&\propto \exp(-x_0^2/2) \exp \left(- (y_0- \rho x_0)^2/(2 \sigma^2)\right) \\
&\propto \exp \left( - \frac{1}{2}\left( x_0^2 + \frac{\rho^2 x_0^2}{\sigma^2}  - \frac{2y_0 \rho x_0}{\sigma^2} \right)  \right) \\
&\propto \exp \left( - \frac{1}{2}\left( \frac{x_0^2}{\sigma ^2}   - \frac{2y_0 \rho x_0}{\sigma^2} \right)  \right) \\
&\propto \exp \left( - \frac{1}{2}\left( \frac{x_0}{\sigma }   -  \frac{y_0 \rho   }{\sigma} \right)^2  \right) \\
&\propto \exp \left( - \frac{1}{2 \sigma ^2}\left( x_0   - y_0 \rho \right)^2  \right) \\
&\propto \exp \left( - \frac{1}{2 \sigma ^2}\left( x_0   - \mu_{y_0} \right)^2  \right), 
 \end{align*}
 which is the mass assigned by $\NDG{\mu_{y_0},\sigma,\theta_{y_0},s'}$;
Here we repeatedly use that $\rho^2 + \sigma^2 = 1$.

\end{proof}

Thus, to prove \Cref{thm:sq-hardness-discrete-gaussian}, it suffices to consider \Cref{def:ngca-lin-regr-discrete},  which is a conditional NGCA instance.
Since the distributions $\{A_y\}_{y \in \R}$ are (necessarily) degenerate, the lower bound machinery of $\SDA$ and pairwise correlations developed in \cite{DiaKS17,DiaKS19} for (conditional) NGCA lead only to vacuous bounds.
To bypass this degeneracy, we will instead  use  \Cref{eq:sq-mid-step-intro} and will show that for any bounded query $f$,
\begin{align}
\label{eq:goal-section-3}
  \P_{v \sim \cS^{d-1}}\{\left|\E_{Z\sim \distDisc_v}[f(Z)] - \E_{Z \sim \distNull}[f(Z)] \right| \geq \text{``large''}\} \leq \text{``tiny''}\,,
\end{align}
where the notion of being ``large'' is according to \Cref{def:notion-of-success} for $m=\widetilde{o}(\rho^2/\sqrt{d})$. 
However, it is unwieldy to compute (or upper bound) this probability. Hence, we first take a detour to a related testing problem with the more usual continuous Gaussian noise in the next section.

\subsection{Conditional NGCA with Continuous Gaussian}
\label{sec:ngca-continuous-noise}
As mentioned in the introduction, we leverage the similarity of discrete Gaussian with continuous Gaussian (with respect to polynomials) as an analysis tool.
We define the analogous testing problem with continuous Gaussian noise below.  
\begin{testingproblem}
\label{def:ngca-lin-regr-discrete-continuous-gaussian}
For $y\in \R$, let $B_y$ denote the distribution  $\cN\left(\mu_y, \sigma^2\right)$ with parameters as in \Cref{parameters}.
Consider \Cref{def:ngca-lin-regr} with
    \begin{itemize}
        \item (Marginal of $y$)  $R := R^*_{\rho,E}$ with $E = \NDG{0,\sigma,0,s}$.
        \item (Conditional NGCA) For each $y\in\R$,  $\hiddenSQ_y$ is equal to $B_y$.
    \end{itemize}
    We denote the corresponding null by $\distNull$ (same as Problem \ref{def:ngca-lin-regr-discrete}) and the alternate for direction $v$ by $\distCont_v$.

\end{testingproblem}
\begin{remark}
Observe that the alternate above $T_v$ does not correspond to the following (Gaussian) linear model: $y= \rho v^\top x + z$ for $x \sim \cN(0,\bI_d)$ and $z \sim \cN(0,\sigma^2)$ independently of $x$. This is because the marginal of $Y$ under the aforementioned linear model would have been Gaussian $\cN(0,1)$, while it is $R$ in \Cref{def:ngca-lin-regr-discrete-continuous-gaussian} (which is not Gaussian).
\end{remark}

Before establishing the connection with discrete Gaussian quantitatively, 
we first establish that \Cref{def:ngca-lin-regr-discrete-continuous-gaussian} is SQ-hard.
In fact, we show the stronger result that the associated SDA is large.
\begin{restatable}[SQ Hardness of Continuous Noise]{proposition}{PropSDAContinuous}
\label{lem:sq-hardness-continuous-noise}
Consider \Cref{def:ngca-lin-regr-discrete-continuous-gaussian}.
Then for any $m \in \N$ and $q \in \N$ satisfying 
$\frac{\rho^2 \sqrt{\log(1/q)}}{\sqrt{d}} \lesssim \frac{1}{m}$, 
we have that $\SDA(m) \gtrsim q$.

\end{restatable}

\begin{proof}

To calculate the average SQ correlation between $\distCont_v$ and $\distCont_{v'}$, we can first calculate the average correlation between the conditional distributions and then take the average marginal over $y$ to obtain the following expression:
\begin{align}
\chi_{\cN(0,\bI_d)\times R} \left(\distCont_v, \distCont_{v'}\right) 
&= \E_{y \sim R}\left[\chi_{\cN(0,\bI_d)}\left(P^{B_{y}}_v, P^{B_{y}}_{v'}\right)\right].
\end{align}
Here, we crucially used that the marginal distribution of $y$ under $P$, $\distCont_v$ and $\distCont_{v'}$ is identical.

Observe that the distribution $P^{\hidCont_y}_v$ is equal to $\cN(\mu_y v, (\bI_d - vv  ^\top) + \sigma ^2 vv^\top)$. 
Using \Cref{lem:chi-gaussian} with $a = \mu_y = \rho y$, $\gamma = 1 - \sigma^2 = \rho^2 $, and $\cos \theta = v^\top v' $ to calculate $\chi_{\cN(0,\bI_d)}\left(P^{B_{y}}_v, P^{B_{y}}_{v'}\right)$, we obtain
\begin{align*}
1 + \chi_{\cN(0,\bI_d)}\left(P^{B_{y}}_v, P^{B_{y}}_{v'}\right)
&= \frac{\exp\left(\frac{\alpha^2 \cos\theta}{1 + \gamma\cos\theta }\right)}{\sqrt{1 - \gamma^2 \cos^2\theta}}= \frac{\exp\left(\frac{\rho^2 y^2 \cos\theta}{1 + \rho^2\cos\theta }\right)}{\sqrt{1 - \gamma^2 \cos^2\theta}} \\
&= \left(1+f(\theta)\right)\exp\left(g(\theta) y^2\right)
\end{align*}
for appropriately defined $f(\theta):= \frac{1}{\sqrt{1 - \rho^4 \cos^2 \theta}} - 1$ and $g(\theta) := \frac{\rho^2 \cos \theta}{1+ \rho^2\cos \theta}$.
Therefore, the average correlation over $y \in R$ is equal to
\begin{align}
\label{eq:explicit-calcualtion-chi-sq}
   \chi_{\cN(0,\bI_d)\times R} \left(\distCont_v, \distCont_{v'}\right) &= \left(1 + f(\theta)\right) \E_{y \sim R}\left[\exp\left(g(\theta) y^2\right)\right] - 1 \,.
\end{align}
Now observe that $|g(\theta)| \leq \rho^2 \leq \rho_0^2$ by assumption for a sufficiently small $\rho_0$. Therefore, if we define  $r(\theta):=\E_{y \sim R}\left[\exp\left(g(\theta) y^2\right)\right] - 1$,
then \Cref{fact:subgaussian-square} and \Cref{lemma:subgaussian-ness} imply that 
\begin{align*}
|r(\theta)|:=\left|\E_{y \sim R}\left[\exp\left(g(\theta) y^2\right)\right] - 1\right| \lesssim |g(\theta)|.
\end{align*}
Combining this with \Cref{eq:explicit-calcualtion-chi-sq}, we obtain
\begin{align*}
\left|\chi_{\cN(0,\bI_d)\times R} \left(\distCont_v, \distCont_{v'}\right)\right| 
&= (1 + f(\theta))(1 + r(\theta)) -1  \lesssim f(\theta) + |r(\theta)| && \tag{using $|r(\theta)|\lesssim1 $} \\
& \lesssim \rho^4 \cos^2 \theta +  \rho^2 |\cos \theta|  \lesssim \rho^2 |\cos \theta|,
\numberthis\label{eq:bound-on-chi-correlation}
\end{align*}
where we use that $\rho^2 |\cos \theta| \leq 0.1$.
In particular,
\begin{align}
\label{eq:chi-square-gaussian}
\chi^2 \left(\distCont_v, \cN(0,\bI_d)\times R\right) \lesssim \rho^2\,.     
\end{align}
We are now ready to show that $\SDA(m) \geq q$, for which we need to show the following:
\begin{align*}
\sup_{\cE: \P_{v,v'}((v,v') \in \cE) \geq 1/q^2} \E_{v,v'}\left[
  \left|\chi_{P} \left(\distCont_v, \distCont_{v'}\right)\right| \big| \cE
\right] \leq  \frac{1}{m}\,.
\end{align*}
Using \Cref{eq:bound-on-chi-correlation}, it suffices to show that 
\begin{align}
\label{eq:SDA-correlation-mid}
\sup_{\cE: \P_{v,v'}((v,v') \in \cE) \geq 1/q^2} \E_{v,v'}\left[
  \left|\rho^2 |v^\top v'|\right| \big| \cE
\right] \leq \frac{1}{m}\,.
\end{align}
If $v$ and $v'$ are two independent random unit vectors, then $W := v^\top v'$ is a centered $\Theta(1/\sqrt{d})$-subgaussian random variable~\cite[Theorem 3.4.6]{Vershynin18}.
Applying \Cref{fact:subgaussian-cond}, we obtain that $\E[|W| \big| \cE] \lesssim (1/\sqrt{d}) \sqrt{\log(1/\P(\cE))} \lesssim \frac{1}{\sqrt{d}} \sqrt{\log(q)}\,$.
Therefore, we have shown that the left hand side in \Cref{eq:SDA-correlation-mid}  is less than $\rho^2\left(\tfrac{O(1)}{\sqrt{d}} \cdot \sqrt{\log q} \right)$ and hence \Cref{eq:SDA-correlation-mid} holds if \begin{align}
\frac{\rho^2 \sqrt{\log(1/q)}}{ \sqrt{d}} \lesssim \frac{1}{m}, \,
\end{align}
which is the desired conclusion.

\end{proof}

As a consequence, \Cref{prop:SDA-upper-bound} implies that for any query $f$ bounded in $[0,1]$ and $m=o(\tfrac{\sqrt{d}}{\rho^2\log^4d})$, it holds that 
\begin{align}
\label{eq:goal-continuous-Gaussian}
  \P_{v \sim \cS^{d-1}}\{\left|\E_{\distCont_v}[f] - \E_{P}[f] \right| \geq \text{``threshold of $\VSTAT(m)$''}\} \leq \tfrac{1}{d^{\omega(\log^2 d)}}\,.
\end{align}

\subsection{Hardness of Distinguishing Discrete and Gaussian Noise}
\label{sec:indist}

Towards establishing \Cref{eq:goal-section-3}, a natural step after proving \Cref{eq:goal-continuous-Gaussian} is to argue that, with high probability, $\left|\E_{\distCont_v}[f(Z)] - \E_{{Q}_v}[f(Z)]\right|$ is small.
This is exactly what we establish in the next result, which is our main technical result:  
\begin{theorem}
\label{prop:distinguish-discrete-vs-cont}
Suppose that (i) $\alpha \gg \frac{1}{d^{\polylog(d)}}$  and $\rho^2 \geq s^2 \log^C(d/\alpha)$ for a large constant $C>0$.  
Then for any $f: \cU \to [0,1]$, it holds that
$\P_{v \sim \cS^{d-1}}\Big[\Big|\E_{\distDisc_v}[f] - \E_{\distCont_v}[f] \Big| \gtrsim \Big(\frac{\alpha}{d}\Big)^{\log^2(d/\alpha)} \Big] \leq \frac{1}{d^{\log^2d}} \,$.
\end{theorem}
In the remainder of this section, we detail the proofs and intuition for the above result.

As a first step, we do a Hermite expansion of the function $f$ as in \cite{DiaKRS23}, but generalized to the setting of conditional NGCA.
However, for technical reasons due to the degeneracy of $A_y$ and hence $\distDisc{}_v$, we would need to perform another truncation operation.
\begin{definition}
Define $\hidDiscT_y$ to be the univariate distribution $A_y$ conditioned on $\{z: |z| \leq d\}$ and let $\distDiscT_v$ to be analogous to $\distDisc_v$ but with $\hidDiscT_y$ instead of $A_y$. 
\end{definition}
We now use the Hermite expansion to obtain the following result:
\begin{proposition}
\label{prop:decomposition-disc-vs-cont-main-body}
Let $f:\R^d\times \R \to [0,1]$. For any  $L \leq [1,\tfrac{d}{2}]$, $\ell\in \N$, and $\widetilde{f}:= f\cdot\1_{|y|\leq L}$, we have
\begin{align}
\label{eq:decomposition-disc-vs-cont-main}
\nonumber
\left|\E_{\distCont_v}[f] - \E_{\distDisc_v}[f]\right| &\lesssim e^{-\Omega(L^2)} + e^{-\Omega(d)} + \littlesum_{k=1}^{\ell}\max_{|y|\leq L}\left|\widetilde{\bA}_{k,y} - \bB_{k,y}\right| \cdot
\E_{y\sim R}\Big[
  \left|\langle v^{\otimes k}, \bT_{k,y}\rangle\right|
\Big]
\\
  & { \qquad + \big|\E_{\distDiscT_v}[\widetilde{f}^{> \ell}] - \E_{\distDisc_v}[\widetilde{f}^{> \ell}]\big|} \;,
\end{align}
where $\bT_{k,y} := \E_{x \sim \cN(0,\bI_d)}[\widetilde{f}_y(x)\bH_k(x)]$,  $\widetilde{\bA}_{k,y}:=\E_{x\sim \hidDiscT_y} [h_k(x)]$, and  $\bB_{k,y}:=\E_{x\sim B_y} [\widetilde{f}(x) ]$.

\end{proposition}
The full proof is deferred to \Cref{sec:proof-prop:decomposition-disc-vs-cont} but we provide a proof sketch below.
\begin{proofsketch}
Since $R_{\rho,E}$ has very light tails,
we can replace $f$ with $\widetilde{f}$ which leads to a difference of at most $\P(|y|\geq L) \lesssim e^{-\Omega(L^2)}$.

Next, we decompose $\widetilde{f}$ as $\widetilde{f}^{\leq \ell}$ and $\widetilde{f}^{> \ell}$, where the $\widetilde{f}^{>\ell}$ term appears as is in \Cref{eq:decomposition-disc-vs-cont-main} and can be ignored momentarily.
Then, using law of total expectation, we can write 
$\E_{(x,y) \sim \distCont_v }[\widetilde{f}^{\leq \ell}] = \E_y[\E_{x}[\widetilde{f}_y^{\leq \ell}(x)]]$. The result in \cite[Lemma 3.3]{DiaKRS23} implies that   $\E_{P^{\hidDisc_y}_v}[\widetilde{f}_y^{\leq \ell}(x)] = \sum_{k=0}^{\ell} \bA_{k,y}\left\langle  v^{\otimes k}, \bT_{k,y}\right\rangle$ for $\bA_{k,y}:=\E_{x\sim \hidDisc_y} [h_k(x)]$.
A similar argument holds for $\hidCont_y$.
Taking the difference and integrating over $y$, we obtain 
    $\E_{\distDisc_v}[\widetilde f^{\le\ell}]
    -\E_{\distCont_v}[\widetilde f^{\le\ell}]
    = E_{y}\bigl[\littlesum_{k=1}^{\ell}
      \bigl(\bA'_{k,y}-\bB_{k,y}\bigr)\,
      \langle v^{\otimes k},\bT_{k,y}\rangle
    \bigr]$.

\looseness=-1Since $\widetilde{f}$ is zero for $|y| \geq L$, $\bT_{k,y}$ is also zero for large $y$ and we can take the maximum only over $|y|\leq L$, yielding \Cref{eq:decomposition-disc-vs-cont-main} roughly. 
However, later on, we would still need to control $\E_{\distDisc}[\widetilde{f}^{> \ell}]$, which could potentially be large because of degeneracy and unboundedness of $\hidDisc_y$s.
Therefore, we replace $\hidDisc_y$s with $\distDiscT_y$s to make it bounded; using  concentration of $\hidDisc_y$s, this leads to an additional $e^{-\Omega(d^2)}$ term.
\end{proofsketch}

Thus, we crucially need to control $|\widetilde{\bA}_{k,y} - \bB_{k,y}|$ and obtain high-probability estimates (over randomness in $v$) on  $\E_{y \sim \R}\left[\left| \langle v^{\otimes k}, \bT_{k,y}\rangle\right|\right]$.
\subsubsection{Closeness of Hermite Coefficients}

We begin with the former goal of controlling $|\widetilde{\bA}_{k,y} - \bB_{k,y}|$.
\begin{restatable}[Closeness of Hermite Coefficients]{lemma}{LemCloseHermiteCoefficients}
\label{lem:Fixed-y-hermite-coefficients}
    For any $y\in \R$ and $k \in \N$, we have:
        \begin{itemize}
      \item (Tighter for small $k$)    \label{eq:hermite-coefficient-small-k}
      $|\widetilde{\bA}_{k,y} - \bB_{k,y}| \lesssim \max\left(1, |\mu_y|^k\right)  k^{O(k)}   \cdot \bigl(e^{- \Omega\left(\frac{\rho^2}{s^2}\right)} + e^{-\Omega(d)}\bigr)$.

    \item (Tighter for larger $k$) $|\widetilde{\bA}_{k,y} - \bB_{k,y}| \lesssim 
    e^{O(\mu_y^2)}$.
    \end{itemize}
    \end{restatable}

 For low-degree Hermite polynomials, the closeness between $\bA_{k,y}$ and $\bB_{k,y}$ (and hence $\widetilde{\bA}_{k,y}$ and $\bB_{k,y}$) would be a consequence of \Cref{fact:polynomial-discrete-vs-gaussian}.
For large-degree Hermite polynomials, we use that Hermite polynomials are bounded by $e^{x^2/4}$ whereas the tails of both $\hidDiscT$ and $\hidCont$ decay slower than $e^{\mu_y^2} e^{-x^2/4}$---a consequence of subgaussianity around $\mu_y$.  The desired conclusion then follows from integrating these bounds. 
We now provide the formal details.

\begin{proof}

We first consider the case for large $k$.

\paragraph{Large $k$.}
For large $k$, we shall use the fact that $|h_k(x)| \leq \exp(x^2/4)$ for all $x \in \R$~\cite{Krasikov2004}.
\Cref{lemma:subgaussian-ness} implies that for both $\hidCont_y$ and $\widetilde{A}_y$,
\begin{align*}
\forall t: \qquad P(|x - \mu_y| \geq t)  \leq O(1)\exp(-x^2/2),
\end{align*}
where we use that $\sigma \leq 1$.
Therefore, under the both $X\sim \widetilde{A}_y$ and $X \sim \hidCont_y$, we have that 
\begin{align}
\label{eq:Concentration-Ay}
\P(|X| > t) \leq O(1) \exp(O(\mu_y^2)) \exp(-0.4t^2)\,.    
\end{align}
Indeed for $t \leq 10 \mu_y$, the upper bound is bigger than $1$ and hence holds; for $t \geq 10 \mu_y$,
$\P(|X|>t) \le \P(|X - \mu_y|\geq 0.9t)\lesssim \exp(-0.4t^2)$.

Therefore, we can upper bound $\E[|h_k(X)|]$ for both distributions as follows:
\begin{align*} 
\E[|h_k(X)|] &\leq \E[e^{X^2/4}] \leq 1 + \int_{1}^\infty \P(|X| > 2 \sqrt{\log_e u})du \lesssim 1+\int_{1}^\infty e^{-0.4 \cdot 4 \cdot \log_e u}\\
& \lesssim \exp(O(\mu_y^2))\left(1 +  \int_1^\infty u^{-1.6}du\right) \lesssim \exp(O(\mu_y^2)).
\end{align*} 

\paragraph{Smaller $k$.}
We first define $\bC_{k,y}:= \E_{x \sim A_y}[h_k(x)]$.
Since $\widetilde{A}_y$ is $A_y$ conditioned on $\cE:= \{z: |z| \leq d\}$ and satisfies $\P(\cE) \geq 1 - \tau$ for $\tau \lesssim e^{-\Omega(d)}$ (see \Cref{lemma:subgaussian-ness}),
we have that for any function $g$:
  \begin{align*}
  \left|\E_{\widetilde{A}_y}[g] - \E_{\hidCont_y}[g]\right| \lesssim 2\left|\E_{A_y}[g] - \E_{\hidCont_y}[g]\right| + \tau \E_{\hidCont_y}[g] + \sqrt{\tau\E_{A_y}[g^2] }  \,.
  \end{align*}
  The above inequality follows by noting that the left hand side above is exactly equal to $\frac{\E_{A_y}[g] - \E_{B_y}[g]}{1 - \tau} + \frac{\tau \E_{B_y}[g]}{1 - \tau} 
  + \frac{\E_{A_y}[g \I_\cE]}{1-\tau}$ and then applying Cauchy-Schwarz inequality.
      In our context, the above display equation yields:
  \begin{align}
    \label{eq:hermite-coeff-diff-prequel}
\left|\bB_{k,y} - \bA_{k,y}\right| \leq  2\left|         \bB_{k,y} - \bC_{k,y}\right| + \tau |\bB_{k,y}| + \sqrt{\tau} \sqrt{\E_{A_y}[h_k^2(x)]}.
  \end{align}
  We will now upper bound this difference.  We first claim that  for $\widetilde{\theta_y} = (\theta_y - \mu_y)/\sigma$ and $\widetilde s= s'
    /\sigma $, we have that
    \begin{align}
    \label{eq:hermite-coeff-diff}
         \bB_{k,y} - \bC_{k,y}  = \E_{x \sim \cN(0,1)}[h_k(\sigma  x + \mu_y)] - \E_{x' \sim \NDG{0,1,\widetilde\theta_y,\widetilde s}}[h_k( \sigma x' + \mu_y)]\,.
    \end{align}
To see this, recall that $\bB_{k,y} = \E_{x \sim \hidCont_y}[h_k(x)] = \E_{x \sim \cN(\mu_y, \sigma ^2)} h_k(x)$, which implies that it is equal to $\E_{x \sim \cN(0,1)} h_k(\sigma x + \mu_y)$. For $\bA_{k,y}$, the claim follows analogously from \Cref{fact:discrete-translation}.

\begin{lemma}
    \label{lem:uncentered-gaussian}
    Let $k \in \N$, $q\in \R$, $a \in \R$, $b \in \R$ and $s'' \ll 1$. Let $G \sim \cN(0,1)$ and $Y \sim \NDG{0,1,q,s'}$. 
\begin{itemize}
  \item $\Big|\E[h_k(b + a G)] - \E[h_k(b + a Y)]\Big| \leq \max(1, |b|^k) \max(1,|a|^k)k^{O(k)} e^{-\frac{1}{s''^2}}\,.$ 
  \item $|\E[|h_k(b + a G)|]|^2 \leq \E[|h_k(b + a G)|^2] \leq k^{O(k)}  \max(1, b^{2k})\max(1, a^{2k})$.
  \item $\E[|h_k(b + a Y)|^2] \leq k^{O(k)}  \max(1, b^{2k})\max(1, a^{2k})$.
\end{itemize}
\end{lemma}
Applying this result on \Cref{eq:hermite-coeff-diff} with $b = \mu_y$, $a = \sigma  \leq 1$ and $s'' = \widetilde{s} = s'/\sigma  =  s/{\rho \sigma }$ and plugging it in \Cref{eq:hermite-coeff-diff-prequel} in combination with $\tau \lesssim e^{-\Omega(d)}$, we get \Cref{eq:hermite-coefficient-small-k}.

\end{proof}

We now provide the proof of \Cref{lem:uncentered-gaussian}
\begin{proof}
Defining the polynomial $p_k(x) := h_k(b + a x)$, we can apply \Cref{fact:polynomial-discrete-vs-gaussian} to $p_k(\cdot)$ to conclude that the deviation in the first item is at most
\begin{align*}
\sqrt{\E_{G \sim \cN(0,1)}[h_k^2(b + aG)]} k! 2^{O(k)} \exp(- \Omega(1/s^2)).
\end{align*}
Hence, to establish both the first and the second items, it remains to show the upper bound 
$\sqrt{\E_{G \sim \cN(0,1)}[h_k^2(b + aG)]} \lesssim k^{O(k)} \max(1, |b|^k)\max(1, |a|^k)$.
To that effect, we use the explicit form of the Hermite polynomials:
$$h_k(x) := \sqrt{k!}\sum_{\ell=0}^{\lfloor k/2\rfloor]} \frac{(-1)^{\ell}}{\ell! (k-2 \ell)! } \frac{1}{2^\ell} x^{k -2 \ell}\,,$$
which gives the following expression:
\begin{align*}
\E[h_k^2(b + a G)] = k!\E\left[\sum_{\ell=0, \ell'=0}^{\lfloor k/2\rfloor]} \frac{(-1)^{\ell}}{\ell! (k-2 \ell)! } \frac{1}{2^\ell} (b+ aG)^{k -2 \ell} \frac{(-1)^{\ell'}}{\ell'! (k-2 \ell')! } \frac{1}{2^{\ell'}} (b+ aG)^{k -2 \ell'}
\right]\,.
\end{align*}
There are  $\Theta(k^2)$ terms in the expression above and 
by linearity of the expectation,
it suffices to control the maximum term above: 
\begin{align*}
\numberthis\label{eq:hermite-higher-power}
\E[h_k^2(b + a G)] &\leq k^2 k! \max_{\ell \leq k/2, \ell' \leq k/2}\E\left[(b+ aG)^{k -2 \ell}  (b+ aG)^{k -2 \ell'}\right] \\
&\leq k^2 k! \max_{\ell \leq k, \ell'\leq k} \sqrt{\E[(b + aG)^{2 \ell}]}\sqrt{\E[(b + aG)^{2 \ell'}]}\\
&\leq k^2 k! \max_{\ell \leq k} \E[(b + aG)^{2 \ell}] \\
&\leq k^2 k! \max_{\ell \leq k} \E[2^{2\ell}b^{2\ell} + 2^{2\ell} a^{2\ell} G^{2\ell}] \\
&\leq 2^{2k}k^2 k! \max_{\ell \leq k} \E[b^{2\ell} + a^{2\ell} G^{2\ell}] \\
&\leq 2^{2k}k^2 k! \max_{\ell \leq k} \E[b^{2\ell} + (O(\sqrt{k}))^k a^{2\ell}] \\
&\leq k^{O(k)}  \max(1, b^{2k})\max(1, a^{2k}),
\end{align*}
which proves the desired result.

We now focus on the third item. Here, we again apply \Cref{fact:polynomial-discrete-vs-gaussian} but this time to the polynomial $p_k^2$, which would then imply that
\begin{align*}
  \E[|h_k(b+aY)|^2] &\leq \E[|h_k(b+ a G)^2] + \sqrt{\E[|h_k(b+ a G)^4]} (2k)^{O(k)}\exp(-\Omega(1/s^2))\\
  &\leq k^{O(k)} \left(\max(1, b^{2k})\max(1, a^{2k}) + \sqrt{\E[|h_k(b+ a G)^4]}\right)  \,.
\end{align*}
To upper bound $\E[|h_k(b+ a G)^4]$, we can use a similar series of arguments as in \Cref{eq:hermite-higher-power}  to get the desired result, wherein we replace the use of Cauchy-Schwarz inequality with the inequality $\E[X_1X_2X_3X_4] \leq \prod_{i=1}^4 (\E[X_i^4])^{1/4}$.
\end{proof}

Thus, \Cref{lem:Fixed-y-hermite-coefficients} implies that (i) for small $k$, the difference is inverse super-polynomially small if $\rho^2 \gg s^2 \polylog(d) \asymp \alpha^2 \polylog(d)$ and (ii) it stays bounded by $O(1)e^{L^2}$ for $|y|\leq L$ for any $k$.

\subsubsection{Small Magnitudes of Random Projections} %

We now turn to computing high-probability estimates on $\E_{y \sim R}|\langle v^{\otimes k}, \bT_{k,y}\rangle|$. Here, we reparameterize the arguments in \cite{DiaKRS23} and obtain the following result, whose proof is deferred to \Cref{app:concentration-of-tensor-dot-unit-vectors}.

\begin{restatable}{proposition}{PropConcVTensor}
    \label{prop:concentration-of-tensor-dot-unit-vectors}
    
    Let $\{\bT_{k,y}\}_{k \in \N, y \in \R}$ be tensors, where for all  $k\in\N$ and $y \in \R$, $\bT_{k,y}$ is a degree-$k$ tensor with $\|\bT_{k,y}\|_2\leq 1$. Let $ t\in \N$ be arbitrary .
    Then for any
     $\delta \in (0,1)$, it holds with probability $1-\delta$ over a random  unit vector $v$ that
    \begin{align*}
        \E_y\left[\littlesum_{k=1}^t | \langle v^{\otimes k}, \bT_{k,y}\rangle|\right] \lesssim t\quad \text{and} \quad \littlesum_{k > t+1}^\infty \E_y\left[|\langle v^{\otimes k}, \bT_{k}\rangle|\right] \lesssim d^{O(1)} \Big(\frac{t \log{\frac{t}{\delta}}}{d}\Big)^{t/4}  +  d^{O(1)}\cdot\tfrac{1}{\delta}e^{- \frac{Cd}{\log\frac{d}{\delta}}}\,.
    \end{align*}

\end{restatable}
The result above is applicable to our setting because for each $y\in \R$: 
$\sum_{k=1}^\infty \|\bT_{k,y}\|_2^2 =  \|\widetilde{f}_y\|_{L_2(\cN(0,\bI_d)}^2 \leq 1$, where the equality uses the orthonormality of Hermite tensors under $\cN(0,\bI_d)$ and the inequality uses that $\widetilde{f}$ is bounded by $1$.

\subsubsection{Proof of \Cref{prop:distinguish-discrete-vs-cont}}

We are now ready to present the proof  of \Cref{prop:distinguish-discrete-vs-cont}.

\begin{proof}[Proof of \Cref{prop:distinguish-discrete-vs-cont}]
Combining \Cref{prop:decomposition-disc-vs-cont-main-body} with \Cref{lem:Fixed-y-hermite-coefficients} and  \Cref{prop:concentration-of-tensor-dot-unit-vectors} and the fact that $|\mu_y| \leq L$ for $L \geq 1$, we obtain that for any $t \in \N$ and $\ell \in \N$ with probability at least $1 - d^{-\log^2 d}$,
\begingroup
  \setlength{\jot}{3pt}
  \begin{align*}
    \bigl|\E_{\distCont_v}[f] - \E_{\distDisc_v}[f]\bigr|
    &\lesssim
      e^{-\Omega(L^2)} + e^{-\Omega(d)}
      + L^t t^{O(t)}\bigl(e^{-\Omega(\rho^2/s^2)}+e^{-\Omega(d)}\bigr)\,t
    \\[-0.2ex]
    &\quad
      + e^{cL^2}(dt)^{O(1)}\Bigl(\tfrac{t\log t\log^3 d}{d}\Bigr)^{t/4}
      + e^{cL^2}\,e^{-d/\polylog(d)}
      + \bigl|\E_{\distCont_v}[\widetilde f^{>\ell}]
      - \E_{\distDiscT_v}[\widetilde f^{>\ell}]\bigr|.
  \end{align*}
\endgroup
For $L = \log^5 d$, $t  = L^6$ and $\rho = s t^2$, the sum of all but the last term is at most $O(e^{-L^2}) \leq d^{-\log^2 (d/\alpha)}$.
For the last term, we show  in \Cref{app:handling_widetilde_f} that taking $\ell$ large enough suffices---this argument uses  \Cref{fact:Hermite} and the truncation of $\hidDisc_y$ as per \cite{DiaKRS23}.
\end{proof}

\subsection{Proof of \Cref{thm:sq-hardness-discrete-gaussian}}
\label{sub:proof-complete-main-testing}

We are now ready to state and prove the formal version of \Cref{thm:sq-hardness-discrete-gaussian}.

\begin{theorem}[SQ Hardness of \Cref{def:lin-regr-oblivious}]
\label{thm:sq-hardness-discrete-gaussian-formal}
Consider the testing problem in \Cref{def:lin-regr-oblivious} with $E = \NDG{0,\sigma,0,s}$ for $s \asymp \alpha$ and $\sigma = 1$. 
Furthermore, assume that
\begin{itemize}
    \item $\alpha \gg \frac{1}{d^{\polylog(d)}}$ (i.e., it is not too tiny)
    \item  $\rho^2 \asymp \alpha^2 \polylog(d/\alpha)$ and $\rho \leq \rho_0$ for a sufficiently small absolute constant $\rho_0$.
\end{itemize}
\noindent Then we have the following guarantees:
\begin{enumerate}
    \item $\P_{E}(z = 0) \geq \alpha$ (i.e., it is a valid instance).
    \item Any SQ algorithm that solves the testing problem with probability at least $2/3$ either uses $q \gtrsim q_0 := d^{\log^2(d/\alpha)}$ many queries or uses a single query which is as powerful as $\VSTAT(m)$ for $m \gtrsim \frac{\sqrt{d}}{\alpha^2 \polylog(d,1/\alpha)}$.
\end{enumerate}
 
 \end{theorem}
\begin{proof}
Since $\sigma \geq 1/2$, we get that $\P(z = 0) \geq \alpha$ (recall that $\P_{Z \sim \NDG{0,\sigma,0,s}}(z = 0) = \Theta(s/\sigma)$), which satisfies the first claim of \Cref{thm:sq-hardness-discrete-gaussian-formal}.

To establish the second claim about the SQ complexity, using \Cref{prop:sq-lower-bound-intro},
it suffices to show that the probability of success of $f$ on distinguishing $\distDisc_v$ and $P$ with $m$ simulation complexity is at most $1/q_0$.
Recall that the success event $\cE_{f,v,m}$ is defined as the following event:
\begin{align}
\label{eq:recall-success}
|\E_{\distDisc_v}[f] - &\E_{P}[f] | \geq \max\Big(\frac{1}{m}, \min\Big(\sqrt{\frac{(\E_P[f])(1 - \E_P[f])}{m}}, \sqrt{\frac{(\E_{\distDisc_v}[f])(1 - \E_{\distDisc_v}[f])}{m}} \Big)\Big)\,,
\end{align}
and our goal is to show that for any fixed bounded query $f: \cZ \to [0,1]$, we have $\P_{v \sim \cS^{d-1}}\left[ \cE_{f,v,m}\right] \leq \frac{1}{q}$ for $q  \asymp d^{\log^2(d/\alpha)}$ and $m  \gtrsim m_0 := \frac{\sigma^2 \sqrt{d}}{\rho^2 \polylog(d/\alpha)}$.

We now define the following events:
\begin{itemize}
  \item  First, $\cE'_{f,v,m}$ is defined as: $\big|\E_{\distDisc_v}[f(x,y)] - \E_{\distCont_v}[f(x,y)] \big| 
\geq \frac{1}{4m^2}$.
\item Next, the event $\cE''_{f,v,m}$ is defined as: for a large constant $C$ (which can be deduced from the proof of \Cref{claim:triangle-ineq}), 
\begin{align*}
\big|\E_{\distCont_v}[f] - &\E_{P}[f] \big|  
&\geq \max\Big(\frac{1}{Cm}, \min\Big(\sqrt{\frac{(\E_P[f])(1 - \E_P[f])}{Cm}}, \sqrt{\frac{(\E_{\distCont_v}[f])(1 - \E_{\distCont_v}[f])}{Cm}} \Big)\Big)\,.
\end{align*}
\end{itemize}
Next, we show in  \Cref{claim:triangle-ineq} that $\cE_{f,v,m} \subset \cE_{f,v,m}' \cup \cE''_{f,v,m}$.
By the union bound and \Cref{claim:triangle-ineq}, it suffices to establish that the probabilities of these events individually is at most $\frac{1}{2q}$.
\begin{itemize}
  \item ($\cE'_{f,v,m}$)  \Cref{prop:distinguish-discrete-vs-cont} implies the desired bound for any $m \leq (d/\alpha)^{\log^2(d/\alpha)}$ and $q\leq d^{\log^2(d/\alpha)}$.

  \item ($\cE'_{f,v,m}$) This inequality was established in \Cref{lem:sq-hardness-continuous-noise} for any $m \lesssim m_0$ with $m_0 \asymp \frac{\sigma^2 \sqrt{d}}{\rho^2 \sqrt{\log(1/q)}}$.
Taking $q = d^{\log^2(d/\alpha)}$ and $\sigma= \Theta(1)$ leads to $m_0 \asymp \frac{\sqrt{d}}{\rho^2 \polylog(d/\alpha)}$.

\end{itemize}
This completes the proof of \Cref{thm:sq-hardness-discrete-gaussian-formal}.
\end{proof}

We now provide the statement and the proof of \Cref{claim:triangle-ineq}.
\begin{claim}
\label{claim:triangle-ineq}
We have that $\cE_{f,v,m} \subset \cE_{f,v,m}' \cup \cE''_{f,v,m}$.
\end{claim}
\begin{proof}[Proof of \Cref{claim:triangle-ineq}]
 Indeed, we have that
\begin{align*}
\big|\E_{Q_v}[f(x,y)] - &\E_{P}[f(x,y)] \big| \\
&\leq \big|\E_{Q_v}[f(x,y)] - \E_{Q'_v}[f(x,y)] \big| + \big|\E_{Q'_v}[f(x,y)] - \E_{P}[f(x,y)] \big| \\
&\leq \frac{1}{4m^2} + \max\left(\frac{1}{Cm}, \min\left(\sqrt{\frac{(\E_P[f])(1 - \E_P[f])}{Cm}}, \sqrt{\frac{(\E_{\distCont_v}[f])(1 - \E_{\distCont_v}[f])}{Cm}} \right)\right)\,.
\end{align*}
Observe that on $\cE_{f,v,m}'$, $|\E_{\distCont_v}[f] - \E_{\distCont_v}[f]| \leq \tau$ for $\tau = O(1/m^2)$.
Since the expectations are close, the standard deviations are also close:
$|\sqrt{(\E_{\distCont_v}[f])(1 - \E_{\distCont_v}[f])} - (\E_{Q_v}[f])(1 - \E_{Q_v}[f])| = O (\sqrt{\tau})$.
Therefore, the second term above is at most
\begin{align*}
\max\left(\frac{1}{Cm}, \min\left(\sqrt{\frac{(\E_P[f])(1 - \E_P[f])}{Cm}}, \sqrt{\frac{(\E_{Q_v}[f])(1 - \E_{Q_v}[f])}{Cm}}  \right)\right) + \sqrt{O(\tau)}\,.
\end{align*}
Since $\tau \gtrsim 1/m^2$,
the overall term is at most 
\begin{align*}
\frac{1}{Cm} + \max\left(\frac{1}{Cm}, \min\left(\sqrt{\frac{(\E_P[f])(1 - \E_P[f])}{Cm}}, \sqrt{\frac{(\E_{Q_v}[f])(1 - \E_{Q_v}[f])}{Cm}}  \right)\right)\,.
\end{align*}
We now claim that this is less than the threshold for $\cE_{f,v,m}$ in \Cref{eq:recall-success}.
Towards that goal, define $a = \sqrt{\E_P[f]\cdot \E_P[1 - f]}$ and $b$ for the corresponding term with $Q_v$.
Consider the case when $\min(a, b)/\sqrt{Cm} \leq 1/{Cm}$.
Then the left hand side above is $\frac{2}{Cm}$, which is less than the quantity in $\cE_{f,v,m}$, which is at least $1/m$.
Suppose now that $\min(a,b) \geq 1/(Cm)$.
Then the term above is at most $\frac{1}{Cm} + \frac{\min(a,b)}{\sqrt{Cm}} \leq 2\frac{\min(a,b)}{Cm}$,
which is less than the quantity in $\cE_{f,v,m}$, which is at least $\frac{\min(a,b)}{m}$.

Thus, we have shown that $\cE_{f,v,m} \subset \cE'_{f,v,m} \cup \cE''_{f,v,m}$.

\end{proof}

%% file: conc.tex
\section{Conclusions and Open Problems} \label{sec:conc}
{In this work, we studied the fundamental 
problem of noiseless linear regression under Gaussian marginals 
with additive oblivious contamination. 
Our main result is an information-computation tradeoff
for SQ algorithms, suggesting that efficient learners 
require sample complexity at least quadratic in $1/\alpha$, 
where $\alpha$ is the fraction of inliers, while 
linear dependence in  $1/\alpha$ 
information-theoretically suffices. 
An immediate open problem concerns the dependence on $d$ 
in the lower bound. Specifically, it is a plausible conjecture that there exists a lower bound of $\Omega(d/\alpha^2)$ on 
the computational sample complexity of the problem (thus, 
exactly matching the sample complexity of known algorithms). 
We note that such a lower bound would require 
a new hardness construction, as our hard testing instance is 
efficiently solvable with $O(d^{1/2}/\alpha^2)$ samples. }

%% file: appendix.tex
\appendix

\begin{center}
    \LARGE \textbf{Appendix}
\end{center}

The Appendix is organized as follows:
\Cref{app:prelims} contains additional preliminaries and background on SQ algorithms.
\Cref{app:proof-deferred} contains proofs deferred from \Cref{sec:generalized_SQ_lower_bound}.
\Cref{app:estimation_is_harder_than_testing} gives a computationally-efficient reduction from testing to estimation.
\Cref{app:inefficient-sq} gives an inefficient SQ algorithm that uses $\VSTAT{}$ oracle with simulation complexity linear in $\tfrac{1}{\alpha}$, whereas \Cref{app:efficient-sq} gives an efficient SQ algorithm that uses a $\VSTAT{}$ oracle with simulation complexity quadratic in $\tfrac{1}{\alpha}$.

\section{Proofs Deferred from Preliminaries}
\label{app:prelims}

\subsection{Useful Technical Facts}

We say a random variable $X$ or a distribution $P$ is $\sigma$-subgaussian if $\P(|X|\geq t) \leq 2\exp(-c t^2/\sigma^2)$ for all $t > 0$; here $c$ is an absolute constant.

\FactSubgaussian*
\begin{proof}
We use expansion of $e^{x}$ and the fact that $\E[|X|^p] \leq (C\sigma \sqrt{p})^p$ for a $\sigma$-subgaussian distribution~\cite[Proposition 2.5.2]{Vershynin18} to get
\begin{align*}
    \E[e^{\frac{aX^2}{\sigma^2}} - 1] =\E\left[ \sum_{i=1}^\infty a^i \frac{X^{2i}}{\sigma^{2i}i!}\right] \leq \sum_{i=1}^\infty a^i \frac{ (c\sigma \sqrt{2i})^{2i}}{\sigma^{2i}i!} \leq \sum_{i=1}^\infty \frac{ ( \sqrt{e a}c \sqrt{2i})^{2i}}{i^i} \leq \sum_{i=1}^\infty (\sqrt{ea}c \sqrt{2})^{2i},  
\end{align*}
which is of order $O(a)$ for $a \leq a_0 = (0.5)/(ec^2 2)$  because it then converges as a geometric sequence.
\end{proof}
For completeness, we provide the proof of \Cref{fact:Hermite}.
\FactHermite*
\begin{proof}
The first statement is a simple consequence of the fact that Hermite polynomials are a complete orthonormal system of $L^2(\R^d, \cN(0,\bI_d))$.

For the second statement, we shall use dominated convergence theorem. 
Define the residue $f_y^{>\ell}(x):= f(x,y)-f_y(x)^{>\ell}$ and $J_{\ell}(y):= \|f_y^{>\ell}\|_{L_2(\cN(0,\bI_d))}^2$. Observe that $\E_{y \sim \R}[J_\ell(y)] = \|f^{>\ell}\|_{L_2(\cN(0,\bI_d) \times R)}^2$.
The first statement implies that for each $y\in\R$, $J_\ell (\cdot) \to 0$ as $\ell \to \infty$.
Furthermore, $J_{\ell}$ is uniformly bounded by $4$ as follows:
\begin{align*}
J_{\ell}(y) = \|f_y(x) - f_y^{\ell}(x)\|_{L_2(\cN(0,\bI_d))}^2 \leq 2 \|f_y\|^2_{L^2(\cN(0,\bI_d))} + 2\|f_y^{\leq \ell}\|^2_{L^2(\cN(0,\bI_d))} \leq 4,
\end{align*}
where we use Parseval's identity to say $\|f_y^{\leq \ell}\|^2_{L^2(\cN(0,\bI_d))} \leq \|f_y\|^2_{L^2(\cN(0,\bI_d))}$ and that $|f_y| \leq $1.
Since $J_\ell \to 0$ pointwise as $\ell \to \infty$ and $0 \leq J_\ell \leq 4$ uniformly, by the dominated convergence theorem, $\E_{y}[J_{\ell}(y)] \to 0$ as $\ell \to \infty$. 

\end{proof}

\subsection{Statistical Query Algorithms}
\label{app:statistical_query_algorithms}

Recall the notion of success from \Cref{def:notion-of-success}.
The notion of success is intimately tied to the $\SDA$ as shown by the following result:
\begin{proposition}[SQ lower bounds using SDA; \Cref{prop:SDA-upper-bound}]
\label{prop:sq-lower-bound-success}
For  any query $f: \cZ \to [0,1]$,
the following holds: $\P_{v \sim \cS^{d-1}}\left( \cE_{f,v.m}  \right) \leq \frac{1}{\SDA(7m)}$.
\end{proposition}
We use the arguments implicit in \cite{FelGRVX17,BreBHLS21,DiaKRS23}.
\begin{proof}
Here, we assume that $\altSQ_v$ has a valid density with respect to $\nullSQ$.  
For a $v \in \cS^{d-1}$ and $z \in \cZ$, we use $q_v(z)$ to denote the Radon–Nikodym derivative of $\altSQ_v$ with respect to $\nullSQ$.
Observe that
\[
  \E_{\altSQ_v}[f] - \E_{\nullSQ}[f]
  = \bigl\langle q_v - 1,\;f\bigr\rangle_{\nullSQ}.
\]

Fix a query $f$ and assume that $a_1 := \E_{\nullSQ}[f] \le \tfrac12$,
otherwise apply the following arguments to $1-f$.
We shall show that $\P(\cE_{f,v,m}) \le 1/\SDA(7m)$ by contradiction.
Suppose $\P(\cE_{f,v,m}) > 1/\SDA(7m)$.

Lemma 3.5 in \cite{FelGRVX17} implies that for any $m\ge1$, $0\le a_1,a_2\le1$,
\[
  \text{if}\quad
  |a_1 - a_2| \;\ge\; 
    \max\!\Bigl(\tfrac1m,\;\min\!\bigl(\sqrt{\tfrac{a_1(1-a_1)}m},\,\sqrt{\tfrac{a_2(1-a_2)}m}\bigr)\Bigr),
  \quad\text{then}\quad
  |a_1 - a_2| \;\ge\; \sqrt{\tfrac{a_1(1-a_1)}{3m}}.
\]
Applying this with $a_1 = \E_{\nullSQ}[f]$ and $a_2 = \E_{\altSQ v}[f]$,
then 
if $f$ succeeds on $v$, then
\[
  \1_{\cE_{f,v,m}}\bigl|\langle q_v - 1,f\rangle_{\nullSQ}\bigr|
  \;\ge\;
  \sqrt{\tfrac{a_1}{6m}}\;\1_{\cE_{f,v,m}}.
\]
Taking expectation over $v$ and squaring gives
\begin{align*}
\P(\cE_{f,v,m})^2  \cdot \frac{a_1}{6 m} &\leq \left(\E_v\left[\1_{\cE_{f,v,m}} \left|\langle q_v - 1, f\rangle_\nullSQ\right| \right]\right)^2 \\
&= \left(\E_v\left[\1_{\cE_{f,v,m}} \langle q_v - 1, f\rangle_\nullSQ \cdot \sign_{\langle q_v - 1, f\rangle_\nullSQ} \right]\right)^2 \\
&= \left(\E_v\left[\E_{\nullSQ}\left[\1_{\cE_{f,v,m}} \left( q_v(Z) - 1\right) \left(f(Z)\right) \cdot \sign_{\langle q_v - 1, f\rangle_\nullSQ} \right]\right]\right)^2 \\
&= \left(\E_\nullSQ\left[\E_{v}\left[\1_{\cE_{f,v,m}} \left( q_v(Z) - 1\right) \left(f(Z)\right) \cdot \sign_{\langle q_v - 1, f\rangle_\nullSQ} \right]\right]\right)^2 \\
&=  \left\langle \E_{v}\left[\1_{\cE_{f,v,m}} \left( q_v - 1\right) \sign_{\langle q_v - 1, f\rangle_\nullSQ} \right] ,  f \right\rangle_\nullSQ ^2 \\
&\leq \left\|  f \right\|_{L_2(\nullSQ)} ^2 \cdot   \left\| \E_{v}\left[\1_{\cE_{f,v,m}} \left( q_v - 1\right) \sign_{\langle q_v - 1, f\rangle_\nullSQ}\right] \right\|_{L_2(\nullSQ)}^2 \\
&= a_1 \cdot  \left( \E_{v,v',\nullSQ} \left[ \1_{\cE_{f,v,m}}\1_{\cE_{f,v',m}} \left( q_v(Z) - 1\right)\left( q_{v'}(Z) - 1\right) \sign_{\langle q_v - 1, f\rangle_\nullSQ} \sign_{\langle q_{v'} - 1, f\rangle_\nullSQ} \right]  \right) \\
&= a_1 \cdot  \left( \E_{v,v'} \left[ \1_{\cE_{f,v,m}}\1_{\cE_{f,v',m}} \left\langle q_v - 1,  q_{v'} - 1\right\rangle_\nullSQ \sign_{\langle q_v - 1, f\rangle_\nullSQ} \sign_{\langle q_{v'} - 1, f\rangle_\nullSQ} \right]  \right) \\
&\leq a_1 \cdot  \left( \E_{v,v'} \left[ \1_{\cE_{f,v,m}}\1_{\cE_{f,v',m}} \left|\left\langle q_v - 1,  q_{v'} - 1\right\rangle_\nullSQ\right| \right]  \right)\,.
\end{align*}
Dividing both sides by $a_1\,\P(\cE_{f,v,m})^2$ and noting independence gives, for 
\(\cE=\cE_{f,v,m}\cap\cE_{f,v',m}\),
\[
  \frac1{6m}
  \;\le\;
  \E_{v,v'}\Bigl[\bigl|\langle q_v - 1,\;q_{v'} - 1\rangle_{\nullSQ}\bigr|\Bigm|\cE\Bigr].
\]
Since $\P(\cE)=\P(\cE_{f,v,m})^2\ge\tfrac{1}{\SDA(7m)^2}$, the definition of $\SDA$ implies the RHS is $<\tfrac{1}{7m}$, a contradiction.
Hence, $\P(\cE_{f,v,m})\leq\tfrac{1}{\SDA(7m)}$.
\end{proof}

\begin{proposition}[Query Complexity Lower Bound; \Cref{eq:sq-mid-step-intro}]
\label{prop:sq-lower-bound}
Fix a $m \in \N$ and $q\in\N$.
Suppose that for all bounded queries $f:\cZ \to [0,1]$, the probability of success is small as follows:
\begin{align}
\label{eq:sq-mid-step}
\P_{v \sim \cS^{d-1}}\left( \cE_{f,v,m}  \right) \leq \frac{1}{q}\,.
\end{align}

Then any (potentially randomized and adaptive to the responses of the previous queries) SQ algorithm $\cA$ for solving \Cref{def:generic-sq} (with failure probability less than $0.25$)  must
use either $\Omega(q) = \Omega(\SDA(7m))$ queries or at least one query as powerful
as $\VSTAT(m + 1)$.
\end{proposition}
Again, we use ideas implicit in  \cite{FelGRVX17,BreBHLS21,DiaKRS23}.
\begin{proof}
We will fix the $\VSTAT$ oracle to be a deterministic oracle $\bV^*$ defined below (independent of the algorithm $\cA$).
Since the oracle is fixed, it suffices to show lower bounds against deterministic algorithms. 

Consider the following oracle $\bV^*$:
\begin{itemize}[leftmargin=2em]
  \item If $\Theta = \nullSQ$, then for any query $f$, it returns  $\E_\nullSQ[f]$.
  \item If $\Theta = \altSQ_v$, then it answers differently based on the ``niceness'' of $f$:
  \begin{itemize}[leftmargin=2em]
    \item (``a good query for $\bV^*$ on $v$'') If $\E_\nullSQ[f]$ is a valid $\VSTAT(m)$ response, then answer $\E_\nullSQ[f]$.
    \item(otherwise)  Return $\E_{\altSQ_v}[f]$.
  \end{itemize}
\end{itemize}
Observe that the oracle $\bV$ above is a valid $\VSTAT(m)$ oracle for both null and alternate.

Now, let $\cA'$ be any deterministic SQ algorithm (deterministic as a function of the answers of the oracle) that solves the testing problem with queries $f_1,\dots,f_{q'}$ for $q' = a q$ for some $a < 1$ to be decided soon.

Consider the case when $\Theta = P$. Recall that the adversary returns $\E_\nullSQ[f_i]$ for $i \in [q']$, which is a valid response.
Then, the accuracy guarantee of $\cA$ implies that $\cA$ must output ``null'' on these instances (because it is deterministic); otherwise the failure probability $\P(\widehat{\Gamma} \neq \Gamma)\geq 0.5$.

Now, consider the alternate case where $\Theta = \altSQ_v$ for $v \sim \cS^{d-1}$.
Observe that if $\cE_{f,v,m}^\complement$ holds for a  query $f$, then it is a  ``a good query for $\bV^*$ on $v$''.
By assumption, the probability (over $v$) that any fixed query $f$ is not good is at most $\P_{\Theta}(\cE_{f,v,m}) \leq 1/q$.
Thus, by a union bound,  the probability (over $v$) that all the queries $\{f_i\}_{i=1}^{q'}$ are good for $\bV^*$ is at least  $1 - a$.
When all the queries are good, the algorithm's input is the same as in the null case, and hence the algorithm must answer ``null''.
Therefore, the overall failure probability is at least $0.5(1-a) \geq 0.25$, which is a contradiction.
\end{proof}

\subsection{Pairwise Correlation}
We will use the following closed-form expression for the pairwise correlations between Gaussians: 
\begin{lemma}
\label{lem:chi-gaussian}
    Let unit vectors $u, v \in \R^d$ and scalars $a\in \R$ and $\gamma \in (0,1)$.
    Let $\cos \theta = u^\top v$. Then
    \begin{align*}
       1+  \chi_{\cN(0,\bI_d)} \left(\cN(av, I - \gamma vv^\top), \cN(au, I - \gamma uu^\top)\right) &=  \frac{\exp\left(\frac{\alpha^2 \cos\theta}{1 + \gamma\cos\theta }\right)}{\sqrt{1 - \gamma^2 \cos^2\theta}}\,.
    \end{align*}
\end{lemma}
\begin{proof}
  For any two Gaussians $A= \cN(\mu_1, \Sigma_1)$ and $B = \cN(\mu_2, \Sigma_2)$, the average correlation with respect to the standard Gaussian is given by the following expression using standard calculations:
    \begin{align}
        \chi_{\cN(0,\bI_d)} (A, B) &= \frac{\exp((h'-h)/2)}{\sqrt{s_1s_2} \sqrt{s_{1,2}}}\,, \qquad \text{where}
    \end{align}
\begin{initemize}
  \item $A= (\Sigma_1^{-1} + \Sigma_2^{-1} - I)^{-1}$
\item $s_1 = \det(\Sigma_1)\qquad\qquad$
\item $s_2 = \det(\Sigma_1)\qquad$
\item $s_{1,2} = \det(A^{-1})\qquad$

  \item $h = \mu_1^\top \Sigma_1^{-1} \mu_1 + \mu_2^\top \Sigma_2^{-1} \mu_2$
  \item $y = \Sigma_1^{-1} \mu_1 + \Sigma_2^{-1} \mu_2 $
  \item $h' = y^\top A y$.
\end{initemize}

We will now instantiate this formula in our context. Towards that goal, we calculate the required expressions below:
\begin{itemize}
    \item (Calculating $s_1$ and $s_2$)  $s_1 = s_2 = 1-\gamma$. Also define $b:= 1 -\gamma$.
\item (Calculating $h$) Moreover, $\Sigma_1^{-1} = (I - vv^\top) + b^{-1}vv^\top$ and $\Sigma_2^{-1} = (I - uu^\top) + b^{-1}uu^\top$.
Therefore, $h = \mu_1^\top \Sigma_1^{-1}\mu_1 + \mu_1^\top \Sigma_1^{-1}\mu_1= 2a^2b^{-1}$.
\item (Calculating $y$) The same calculations as above give $y =\Sigma_1^{-1}\mu_1 + \Sigma_1^{-1}\mu_1 = ab^{-1}(v+u)$.
\item (Calculating $s_{1,2}$)
We begin by calculating $A^{-1}$:
\begin{align*}
    A^{-1} &= \Sigma_1^{-1} + \Sigma_2^{-1} - I = (I - vv^\top) + b^{-1}vv^\top + (I - uu^\top) + b^{-1}uu^\top - I \\
    &= I + \frac{\gamma}{1-\gamma} \cdot ( vv^\top +  uu^\top)\,.
\end{align*}
Therefore, the determinant of $A^{-1}$ is $\frac{1- \gamma^2 \cos^2\theta}{(1-\gamma)^2}$
for $\alpha = \tfrac{\gamma}{1-\gamma}$.
This can be seen as follows by considering $2\times 2$ matrices:
\begin{align*}
\det \left(I + \begin{bmatrix}
        \alpha & \alpha u^\top v\\
        \alpha u^\top v & \alpha
    \end{bmatrix} \right) = \det \left( \begin{bmatrix}
        1 + \alpha & \alpha u^\top v\\
        \alpha u^\top v & 1 + \alpha
    \end{bmatrix} \right) 
    = (1+\alpha)^2 - \alpha^2 (u^\top v)^2, 
\end{align*}
which equals the expression above.
\item (Calculating $A$ and $h'$)
Letting $U = [u ; v] \in \R^{d \times 2}$ and $C = \alpha \bI_2$, 
then
\begin{align*}
   A = (I +  UC U^\top)^{-1}&= I - U(C^{-1} + U^\top U )U^\top \\
    &= I - U\left(\begin{bmatrix}
        \alpha^{-1} & 0 \\0 & \alpha^{-1}
    \end{bmatrix}
    + \begin{bmatrix}
        1 & \cos \theta \\ \cos \theta & 1
    \end{bmatrix}\right)^{-1}U^\top\\
    &= I - U\left(\begin{bmatrix}
        1 + \alpha^{-1} & \cos \theta \\\cos \theta & 1 +\alpha^{-1}
    \end{bmatrix}\right)^{-1} U^\top\\
    &= I - U\left(\begin{bmatrix}
        \gamma^{-1} & \cos \theta \\ \cos \theta & \gamma^{-1}
    \end{bmatrix}\right)^{-1} U^\top\\
    &= I - \frac{1}{1 -\gamma^2 \cos^2\theta} U\left(\begin{bmatrix}
        \gamma & -\gamma^2 \cos \theta \\ \gamma^2\cos\theta & \gamma
    \end{bmatrix}\right)^{-1} U^\top\\
    &= I - \left( \frac{\gamma}{1 - \gamma^2\cos^2\theta} \cdot (uu^\top + vv^\top) + \frac{-\gamma^2\cos\theta}{1 - \gamma^2\cos^2\theta} (vu^\top + uv^\top) \right)\\
    &= I - \frac{\gamma}{1 - \gamma^2\cos^2\theta} \cdot (uu^\top + vv^\top) + \frac{\gamma^2\cos\theta}{1 - \gamma^2\cos^2\theta} (vu^\top + uv^\top)\,.
\end{align*}
First, observe that  $(u+v)^\top \bJ uu^\top (u+v) = (1+ \cos\theta)^2$ for $\bJ\in\{uu^\top,vv^\top, vu^\top,uv^\top\}$.
Therefore, $(u+v)^\top M (u+v)$ equals for $M:= I-A$
\begin{align*}
   2 (1+\cos\theta)^2\cdot \frac{1}{1-\gamma^2\cos^2\theta} \cdot (\gamma - \gamma^2\cos \theta) = \frac{\gamma(1+\cos\theta)^2}{1 + \gamma \cos\theta}\,.
\end{align*}

Let $M$ be $I-A$. Then 
\begin{align*}
    h &= a^2b^{-2}y^\top A y = a^2b^{-2}\left(u+v\right)^\top (I - M) \left(u+v\right)  \\
    &=a^2b^{-2}\left(2 + 2\cos \theta  - \frac{2\gamma(1+\cos\theta)^2}{1 + \gamma \cos\theta}\right)\\
    &=2a^2b^{-2}(1+\cos\theta)\left(1  - \frac{(1+\cos\theta)\gamma}{1+\gamma\cos\theta}\right)\\
    &=2a^2b^{-2}(1+\cos\theta)\left(\frac{(1 - \gamma)}{1+\gamma\cos\theta}\right) = 2a^2b^{-1}\frac{(1+\cos\theta)}{(1+\gamma \cos\theta)}\,.
\end{align*}
\end{itemize}
Overall, we get that $s_1s_2s_{1,2} = 1 - \gamma^2\cos^2\theta$
and $h'-h = 2a^2b^{-1}\frac{(1-\gamma) \cos\theta}{1+\gamma \cos\theta} = \frac{2a^2\cos\theta}{1+\gamma \cos \theta} $, which completes the proof.
\end{proof}

\subsection{Basic Fact on Discrete Gaussians}

\FactDiscreteGaussian*
\begin{proof}
First the support of both $X$ and $X'$ are equal to $\theta + s \Z$ (indeed the support of $Y$ is $\theta' + s' \Z$, which when multiplied by $\sigma$, yields $(\theta - \mu) + s\Z $, and further shifting by $\mu$ yields $\theta + s \Z$.

Starting with $X$, for any $i \in \Z$, $\P(X = \theta + s i) \propto s \phi_{\mu, \sigma}(\theta + si) \propto \frac{s}{\sigma}
 \exp(- 0.5( \theta + si  -\mu)^2/\sigma^2) \propto  \exp(- 0.5( \theta + si  -\mu)^2/\sigma^2)$, where we use that $s$ and $\sigma$ can be absorbed into the normalizing constant.

Turning to the random variable $X'$, 
\begin{align*}
\P(X' = \theta + si) &= \P(\sigma Y + \mu = \theta + si) = \P(Y = (\theta- \mu)/\sigma + si/\sigma ) = \P(Y = \theta' + s'i) \\
&\propto \exp(-0.5(\theta' + s'i)^2) \propto \exp\left(-0.5 \left( \frac{\theta + si - \mu}{\sigma}  \right)^2\right),
\end{align*}
where we used the definitions of $\theta'$ and $s'$.
Since the support is equal and the two distributions are equal up to a normalization constant, they must be equal.

\end{proof}

\section{Proofs Deferred from \Cref{sec:generalized_SQ_lower_bound}}
\label{app:proof-deferred}

\subsection{Concentration Properties of Distributions}
\label{app:concentration_properties_of_distributions}
In this section, we state the concentration properties of various distributions that appear in our analysis.
\begin{restatable}{lemma}{LemConcProp}
\label{lemma:subgaussian-ness}
Let the parameters be as in \Cref{parameters}.
Then we have the following:
\begin{enumerate}
\item The distributions $A_y$, $\hidDiscT_y$, and $B_y$ are $O(|y|+\sigma)$-subgaussian and if $X$ follows either one of these distributions, then $\P(|X-\mu_y|>t) \lesssim e^{-t^2/2\sigma^2}$.
  \item The distribution $\NDG{0, \sigma, 0, s}$ is an $O(\sigma)$-subgaussian distribution, and  $R_{\rho,E}$ is an $O(\sigma + \rho)$-subgaussian distribution.
  \item The distributions $P^{A_y}_v$, $P^{\hidDiscT_y}_v$,  and $P^{\hidCont_y}_v$ are  a $O(|y| + \sigma+1)$-subgaussian  distributions.\footnote{A multivariate random vector $X$ is termed $\sigma$-subgaussian if, for all unit vectors $v$, the real-valued random variable $v^\top X$ is $\sigma$-subgaussian.}
\end{enumerate}
\end{restatable}
\begin{proof}
We do it case-by-case.
\begin{enumerate}[wide]
\item For $s'':= s'/\sigma$  and that $t \geq 1$:
\begin{align*}
\P_{X \sim A_y}&(|x - \mu_y|> t)\\
&= \P_{X \sim \NDG{0,1, (\theta_y-\mu_y)/\sigma , s''}}(|W|>  t/\sigma) && (\text{\Cref{fact:discrete-translation}})
\\
&\lesssim  \sum_{i\in \N} s'' \phi_{0,1}(\tfrac{t}{\sigma}+ s''i) + \sum_{i\in \N} s \phi_{0,1}(- \tfrac{t}{\sigma}-s''i)  && (\text{\Cref{fact:polynomial-discrete-vs-gaussian} as $s''\ll 1$})\\ 
&\lesssim \sum_{i \in \N}s'' e^{-0.5\sigma^{-2} t^2-0.5 s''^2 i^2}&&\text{($\sigma \leq 1$)}\\
&\lesssim e^{-t^2/2\sigma^2}\sum_{i \in \N}s'' e^{-0.5i^2 s''^2 }\\
&\lesssim e^{-t^2/2\sigma^2}\,.
\numberthis\label{eq:tailbound-AY}
\end{align*}
    This tail also implies $O(|y| + \sigma)$-subgaussianity as follows: we claim that $P(|X|>t)\lesssim e^{-\frac{c t^2}{\max(|\mu_y|,\sigma)^2}}$. Observe that it suffices to consider $t\gtrsim \max(|\mu_y|,\sigma)$; otherwise, the bound is trivially true.
    For $t\gg |\mu_y|$, $\P(|X|>t)\leq \P(|X-\mu_y|\geq t/2)$ and we can then use \Cref{eq:tailbound-AY}.
    The same arguments hold for $B_y$.
    The claim for the tails of $|X|$ under $\hidDiscT_y$ follows from that of $A_y$ because $\hidDiscT_y$ is obtained from conditioning on an event of probability at least $0.5$.

    \item The claim for $\NDG{0, \sigma, 0, s}$ follows from \Cref{eq:tailbound-AY}. For $R_{\rho,E}$, we use the fact that if $x_1$ and $x_2$ are two independent $\sigma_1$ and $\sigma_2$-subgaussian random variables, then their sum is $O(\sqrt{\sigma_1 + \sigma_2})$-subgaussian~\cite[Proposition 2.6.1]{Vershynin18}.
    \item After rotating appropriately, $P_v^{A_y}$ and $P^{\hidCont}_y$ are vectors of independent coordinates and thus follow a multivariate subgaussian distribution with variance proxy bounded by the subgaussian parameter of any individual coordinate \cite[Lemma 3.4.2]{Vershynin18}. The subgaussian proxy for the $v$ direction is established in the first item, while for the other coordinates it is $O(1)$.
\end{enumerate}
\end{proof}

\subsection{Proof of \Cref{prop:decomposition-disc-vs-cont-main-body}}
\label{sec:proof-prop:decomposition-disc-vs-cont}
Observe that  \Cref{prop:decomposition-disc-vs-cont-main-body} follows from the result below because of 
\Cref{lemma:subgaussian-ness}.
Indeed, \Cref{lemma:subgaussian-ness} implies that (i)  $\P_{y}(|y|\geq L) \lesssim e^{-\Omega(L^2)}$ as $\sigma \lesssim 1$ and $\rho \lesssim 1$
and (ii) for any $y$ with $|y| \leq d/2$, $\P_{A_y}(|z|>d) \leq \P_{A_y}(|z - \mu_y|>d/2)\lesssim e^{-\Omega(d^2)}$.

\begin{proposition}
\label{prop:decomposition-disc-vs-cont}
Let $f:\cU \to [0,1]$. For $L \geq 1$,  define the set $\cC: \{y: |y| \leq L \widetilde\sigma \}$ and the function $\widetilde{f} := f \1_{y \in \cC}$.
Then for any $\ell \in \N$, we have that for $y \sim R_{\rho,E}$:
\begin{align}
\label{eq:decomposition-disc-vs-cont}
\nonumber
\left|\E_{\distCont_v}[f] - \E_{\distDisc_v}[f]\right| &\leq {4\P(y \not \in \cC )} + \max_{y:|y|\leq L} \P_{A_y}(|z|>d)  + \Big|
\E_{y}\Big[
\littlesum_{k=1}^{\ell}   \left( \widetilde{\bA}_{k,y} - \bB_{k,y} \right)\langle v^{\otimes k}, \bT_{k,y}\rangle
\Big]\Big| \\
  & { + \left|\E_{\distDiscT_v}[\widetilde{f}^{> \ell}] - \E_{\distDisc_v}[\widetilde{f}^{> \ell}]\right|}
\,,
\end{align}
where $\bT_{k,y} := \E_{x \sim \cN(0,\bI_d)}[\widetilde{f}_y(x)H_k(x)]$,  $\bA_{k,y}:=\E_{x\sim \widetilde{A}_y} [h_k(x)]$ and $\bB_{k,y}:=\E_{x\sim \hidCont_y} [\widetilde{f}(x) ]$ for $A_y, \hidCont_y$ defined in \Cref{def:ngca-lin-regr-discrete,def:ngca-lin-regr-discrete-continuous-gaussian}.

\end{proposition}
\label{app:proof-decomposition-disc-vs-cont}
\begin{proof}
We start by replacing $\E_{\distDisc_v}[f]$ with $\E_{\distDiscT_v}[f]$ at the cost of additive $\dtv(\distDisc_v,\distDiscT_v)$.
This total variation distance is $O(\P_{\distDisc_v}(|v^\top x|> d))$, which can be upper bounded by $\P(|y|\geq L) + \max_{y:|y|\leq L} \P(|v^\top x| \geq d)$.
Hence, in the rest of this proof, we shall use $\distDiscT$ everywhere.

Next we decompose $f = \widetilde{f} + f'$ for $f' := f \1_{y \not \in \cC}$.

Then we further decompose $\widetilde{f}$ as $\widetilde{f}^{\leq \ell} + \widetilde{f}_y^{> \ell}$.
By triangle inequality, it suffices to show that the expectations of $\widetilde{f}^{\leq \ell}, \widetilde{f}^{> \ell)}$, and  $f'$ are close.
Observe that the term for $\widetilde{f}^{>\ell}$ is already present in the final conclusion. 
Next, for $f'$, the boundedness of $f$ and the same marginals of $Q_v$ and $\distCont_v$ imply that 
\begin{align*}
    \left|\E_{\distDiscT_v}[f'] - \E_{\distCont_v}[f']\right| \leq  2 \P(y \not\in\cC).
\end{align*}

In the remainder, we focus on the terms corresponding to $\widetilde{f}^{\leq \ell}$. 
By the law of total expectation (whose validity for $\widetilde{f}^{\leq \ell}$ is justified below), we have that
\begin{align}
\label{eq:conditional-expectation}
\E_{\distDiscT_v}[\widetilde{f}_y^{\leq \ell}(x,y)] = \E_{y}\left[\E_{x \sim P^{A_y}_v} \left[\widetilde{f}_y^{\leq \ell}\right]\right]\,.
\end{align}

To compute the inner expectation, which is an instance of the unsupervised NGCA, 
we will use \cite[Lemma 3.3]{DiaKRS23}:
\begin{lemma}[Fourier Decomposition Lemma of \cite{DiaKRS23}] \label{lem:hermite-decomposition}
Let $A'$ be any distribution supported on $\R$ and $v$ a unit vector.
Then for any $\ell\in \N$ and $g:\R^d \to [0,1]$, 
\begin{align*}
\E_{x\sim P^{A'}_{v}} [g^{\leq \ell}(x)]&=
\sum_{k=0}^\ell \bA_k \left\langle v^{\otimes k}, \bT_k\right\rangle\,,
\end{align*}
where $\bA_k=\E_{x\sim A'} [h_k(x)]$ and $\bT_k=\E_{x\sim \cN(0,\bI_d)} [g(x) \bH_k(x)]$. 
\end{lemma}

Consider a fixed $y_0 \in \R$ and apply the above result to $A' := \hidDiscT_{y_0}$ and $g(x) := \widetilde{f}_y(x) := f(x, y_0) \1_{y_0 \in \cC}$.
Define $\widetilde{\bA}_{k,y}:= \E_{x\sim \hidDiscT_y} [h_k(x)]$ and $\bB_{k,y}:= \E_{x\sim \hidCont_y} [h_k(x)]$
and $\bT_{k,y}:=\E_{x\sim \cN(0,\bI_d)} [\widetilde{f}(x,y) \bH_k(x)]$. 
We obtain that 
\begin{align*}
\E_{\distDiscT_v}[\widetilde{f}^{\leq \cL}] &= \E_{y}\left[\E_{x \sim P^{\widetilde{A}_y}_v} \left[\widetilde{f}_y^{\leq \ell}(x)\right]\right]
= \E_{y }\left[
\sum_{k=0}^{\ell} \widetilde{\bA}_{k,y}\left\langle  v^{\otimes k}, \bT_{k,y}\right\rangle
\right]\,.
\end{align*}
Observe that the term $k=0$ corresponds to $\E_{\cN(0,\bI_d)}[f(x,y_0)]$ for each $y_0$, implying that the expectation of the $k=0$ term (over $y$) is exactly $\E_{P}[\widetilde{f}(x,y)]$.
Thus,  we get the following decomposition:
\begin{align}
\label{eq:decomposition-discrete}
\E_{\distDiscT_v}[\widetilde{f}^{\leq \cL}] - \E_{P}[\widetilde{f}] &= \E_{y \sim R'}\left[
\sum_{k=1}^\ell \widetilde\bA_{k,y}\left\langle v^{\otimes k}, \bT_{k,y}\right\rangle
\right]\,.
\end{align}
Similarly, the decomposition for the continuous Gaussian noise is as follows:
\begin{align}
\label{eq:decomposition-continuous}
\E_{\distCont_v}[\widetilde{f}_y^{\leq \ell}(x)] - \E_{P}[\widetilde f] &= \E_{y \sim R'}\left[
\sum_{k=1}^\ell  \bB_{k,y}\left\langle v^{\otimes k}, \bT_{k,y}\right\rangle
\right]\,.
\end{align}
The claim follows by combining \Cref{eq:decomposition-discrete,eq:decomposition-continuous}. 

\paragraph{Justifying \Cref{eq:conditional-expectation}.}
It suffices to show that $\E_{\distDiscT_v}[|\widetilde{f}^{\leq \ell}|]< \infty$, which we will establish below.
By Fubini's theorem, we have that 
\begin{align*}
\E_{\distDiscT_v}[|\widetilde{f}^{\leq \ell}|]&= \E_{y}\left[\E_{x \sim P^{\hidDiscT_y}_v} \left[\left|\widetilde{f}_y^{\leq \ell}\right|\right]\right] \\
&\leq \E_{y}\left[\E_{x \sim P^{\hidDiscT_y}_v} \left[\left| \sum_{k=0}^\ell\langle [\E_{x' \sim P^{\hidDiscT_y}_v} f(x')\bH_{k}(x')] , \bH_{k}(x)\rangle\right|\right]\right] \\
&\leq \sum_{k=0}^\ell \E_{y}\left[\E_{x \sim P^{\hidDiscT_y}_v} \left[\left|\langle [\E_{x' \sim P^{\hidDiscT_y}_v} f(x')\bH_{k}(x')] , \bH_{k}(x)\rangle\right|\right]\right].
\end{align*}
This can be further upper bounded by finite sum of the terms (at most $d^{O(\ell)}$) involving 
$$\E_y\E_{x \sim P^{\hidDiscT_y}_v}[| \E_{x'} f(x')p(x')]\cdot | p(x)|]$$
for some polynomials $p(\cdot)$.
Since $|f|$ is upper bounded by $1$, the term above is further upper bounded by 
 $\E_y\E_{x \sim P^{\hidDiscT_y}_v}[|p(x)|^2]$ using Jensen's inequality. 
 Using \Cref{lemma:subgaussian-ness}, $\E_{x \sim P^{\hidDiscT_y}_v}[|p(x)|^2]$ is upper bounded by $\poly(|\mu_y|, d, \|p\|_{\ell_2})$ and since $\mu_y$ is linear in $y$, $\E[\poly(\mu_y)]$ is also finite because $R$ is $O(1)$-subgaussian.

A similar reasoning justifies \Cref{eq:conditional-expectation} for $\distCont_v$.
\end{proof}

\subsection{Proof of \Cref{prop:concentration-of-tensor-dot-unit-vectors}}
\label{app:concentration-of-tensor-dot-unit-vectors}

\PropConcVTensor*
  
The result for $k \leq t$ follows by the claim that $\bT_{k,y}$ has norm at most $1$ almost surely. 
Hence, we will focus on the second claim, for which we shall crucially use the concentration results from \cite{DiaKRS23}, which we state in a different formulation below.

\begin{lemma}[Lemma 3.7 and Corollary 3.9 in \cite{DiaKRS23}]
        \label{lem:conc-tensor-random-dir}
        Let $\bT_{k,y}$ be a random $k$-tensor supported with randomness $y$.
        For a random unit vector $v$ independent of $y$, let $W_k$ denote the random variable $\E_y\left|\langle v^{\otimes k}, \bT_{k,y} \rangle\right|$.
        Let $W'$ be the random variable $v^\top w$ for a unit vector $w \in \cS^{d-1}$.
        Then for any even $p \in \N$,
        \begin{align}
        \label{eq:DKRS23}
            \|W_k\|_{L_p} \leq (\E_y\|\bT_{k,y}\|_2^p)^{1/p} \|W'\|_{L_{pk/2}}^{k/2}.
        \end{align}
        In particular, if $\langle \bT_{k,y},\bT_{k,y}\rangle \leq 1$ almost surely, then there exists a constant $C>0$ such that the following conclusion holds for any even $p \in \N$ and $k \in \N$:
        \begin{enumerate}
            \item $\|W_k\|_{L_p} \leq \left(\frac{Cpk}{d}\right)^{k/4} $. {\hfill \grey (useful for moderate $k$: $k = o(d)$)}
            \item  $\|W_k\|_{L_p} \lesssim \min(d,k)^{1/p} \exp\left(-C  \frac{kd}{\max(d,pk)}\right).$ {\hfill\grey (useful for large $k$: $k \asymp d$)}
            \item $\|W_k\|_{L_p} \leq C \left(\frac{d}{pk}\right)^{d/p}$. {\hfill \grey (useful for extremely large $k$: $k = \omega(d)$)}
        \end{enumerate}
\end{lemma}
While \Cref{eq:DKRS23} is established in \cite[Lemma 3.7]{DiaKRS23} for a fixed tensor $\bT$, the desired follows by  Jensen's inequality: for any even $p$ and error bound $g(v,y)$, we have that $\E_v(\E_y[g(v,y)])]^p \leq \E_v \E_y [g(v,y)^p] = \E_y(\E_v [g(v,y)^p])$, where one can now use \cite[Lemma 3.7]{DiaKRS23}.

We now couple \Cref{lem:conc-tensor-random-dir} with the simple fact that for any random variable $X$,
with probability $1-\delta$,
    $|X| \leq (1/\delta)^{1/p}\|X\|_{L_{p}}$ for any $p \geq 1$.
Therefore, for any $k \geq t$, with probability $1 - \delta/k^2 $,
$|W_k|$ is less than $ c \min(\|W_k\|_{L_{\log(k/\delta)}},  \sqrt{k/\delta} \|W_k\|_{L_2} )$.
We now calculate these bounds separately for different $k$.

\begin{itemize}
  \item (Small $k$ and large $p$) Define $p_k \asymp \log (k/\delta)$.  For any $k$ such that $k \geq t$ and $k p_k \leq C'd p_k$, we have
\begin{align*}
\|W_k\|_{L_{p_k}} \leq     \left(\frac{Ct \log{(t/\delta)}}{d}\right)^{t/4}  +  \exp\left(-C \frac{d}{\log^2(d/\delta)}   \right) \,.
\end{align*} 
This follows by considering the following two regimes separately:
\begin{itemize}
    \item ($c'p_kk \leq d$ for a tiny enough constant $c'$) In this regime, the bound $\left(\frac{Ckp_k}{d}\right)^{k/4}$ is decreasing in $k$ and thus the maximum is achieved at $k=t$.
    \item ($ C'd p_k \geq c'p_kk \geq  d$ for a large constant $c'$) In this regime, the second bound gives the desired result by noting $\max(d,p_kk) \leq C' d_pk$  and $k \geq d/p_k \geq d/p_d$.
\end{itemize}

\item (Large $k$ and $p=2$) Moreover, if $k \geq C'd$, then $\|W_k\|_2 \leq (\frac{d}{2k})^{d/2}$.
Therefore, with probability $1 - \delta/k^{2}$, $|W_k| \leq \sqrt{(k^{2}/\delta)} (d/2k)^{d/2} \leq \sqrt{d/\delta} (d/2k)^{d/4}$.

\end{itemize}

Taking a union bound, we obtain the following bound that holds with probability at least $1-\delta$,
\begin{align*}
    \sum_{k> t} |W_k| &= \sum_{k \in [t, C'd]} k^a |W_k| + \sum_{k \in [t, C'p_k]} k^a |W_k|\\
    &\lesssim (C'd)^{1} \cdot \left( \left(\frac{t \log{(t/\delta)}}{d}\right)^{t/4}  +  d\exp\left(-C \frac{d}{\log(d/\delta)} \right)\right) + \sum_{k \geq C'd} \sqrt{d/\delta} \left(\frac{d}{2k}\right)^{d/2}\,.
    \end{align*}
The summation $\sum_{k \geq C'd}\left(\frac{d}{2k}\right)^{d/2}$ can be upper bounded by a constant factor multiple of the first expression in the sum (this can be seen by integrating $\int_{x \geq x_0} x^{-a} dx$ for $a>2$), and the first expression is at most $e^{-\Omega(d)}$ because $C'$ is large enough.

\subsection{Handling $\widetilde{f}^{> \ell}$}
\label{app:handling_widetilde_f}

We now show that for any $f:\R^d \times \R \to [0,1]$ and any $\delta \in (0,1)$, there exists $\ell \in \N$, depending only on $(f, d, \delta, L,\sigma, \rho, s, \alpha)$
 such that with $1 -\delta$, $|\E_{\altSQ_v}[\widetilde{f}^{>\ell}] - \E_{\distCont_v}[\widetilde{f}^{>\ell}]|$ is smaller than $\gamma$ for a $\gamma$ appropriately small.

First by \Cref{fact:Hermite}, we know that there exists an $\ell(\gamma')$
so that $\|\widetilde{f}^{>\ell}\|_{L_2(P)} \leq \gamma'$.
Since $\chi^2(P,\distCont_v) $ is finite (as established in \Cref{eq:chi-square-gaussian}), 
this implies that $\|\widetilde{f}^{>\ell}\|_{L_2(\distCont_v)}$ is also sufficiently small.
By Cauchy-Schwarz, we get that for every $\gamma$ and $v \in \cS^{d-1}$, there exists an $\ell'(\delta,d,\gamma)$ so that 
$|\E_{\distCont_v}[\widetilde{f}^{>\ell}]| \leq \gamma$. 

Thus, it remains to argue about $\E_{\altSQ_v}[\widetilde{f}^{>\ell}]$.
By a Markov inequality, it suffices to show that $\E_v[|\E_{\altSQ_v}[\widetilde{f}^{>\ell}]|] \leq \E_v\E_{Q_v}[|\widetilde{f}^{>\ell}|] \leq \gamma/\delta$.
Let $D$ be the distribution of $(x,y)$ obtained over $(x,y)$ as follows:
first $y \sim R$ and then $v \sim \cS^{d-1}$ and $x \sim P^{\widetilde{A}_y}_v$.
Let $D_y$ be the conditional distribution of $x$ given $y$ under $D$.
Thus $\E_v\E_{Q_v}[|\widetilde{f}^{>\ell}|] \leq \gamma/\delta= \E_D[|\widetilde{f}^{>\ell}|]$.
We will now show that $\chi^2(D,P) < \infty$, which would suffice for our result. 
Observe that 
\begin{align*}
\chi^2(D,P) :&= \int_{y\sim R} \int_{x} \frac{R^2(y)D_y^2(x)}{R(y) G(x)} dx dy \\
&= \int_{y \sim R} R(y)\int_{x} \frac{D_y^2(x)}{G(x) dx} = \int_{y \sim R} R(y) \chi^2(D_y, \cN(0,\bI_d))\,. 
\end{align*}
Observe that $D_y$ is obtained from $P^{\widetilde{A}_y}_v$ where $\widetilde{A}$ is supported only on $\{x: |x| \leq d\}$.
\cite[Lemma 3.1]{DiaKRS23} implies that $\chi^2(D_v, \cN(0,\bI_d))$ is uniformly upper bounded by $O_d(1)$.
Integrating this uniform upper bound by $O_d(1)$, we get the desired conclusion of $\chi^2(D, P)< \infty$.

\section{Computationally-Efficient Reduction from Testing to Estimation} %
\label{app:estimation_is_harder_than_testing}

Suppose there is an algorithm $\cA$ with the following guarantees:
given $n$ i.i.d.\ samples $(x_1,y_1),\dots, (x_n,y_n)$ in $\R^d \times \R$ from \Cref{def:estimation-problem} with inlier probability $\alpha$ and regressor $\beta \in \R^d$,
computes an estimate $\widehat{\beta}$ such that $\|\widehat{\beta} - \beta\|_2 \lesssim \tau$.

Consider the following  (randomized) algorithm $\cA'$ that takes $2n$ samples $S= \{(x_1,y_1),\dots, (x_n,y_n)\}$ and 
$S'=\{(x_1',y_1'),\dots, (x_n',y_n')\}$
and perform the following operation:
    \begin{itemize}
  \item Sample a random rotation matrix $\bU \in \R^{d \times d}$.
  \item Let $\widehat{\beta}_1$ be the output of $\cA$ on $S$.
  \item Define $S'' := \{(\bU x_1',y_1'),\dots, ( \bU x_n',y_n')\}$
  \item  Let $ w $ be the output  of $\cA$ on $S''$.
  \item Let $\widehat{\beta}_2 = \bU^\top w$.
  \item Let $W = \left \langle \frac{\widehat{\beta}_1}{\|\widehat{\beta}_1\|}, \frac{\widehat{\beta}_2}{\|\widehat{\beta}_2\|_2} \right \rangle $. If $ |W|  > 1/9$, output ``alternate'', otherwise output ``null''.
\end{itemize}

\begin{theorem}
If $\tau \leq  \rho/4$ and $d \gtrsim \log(1/\delta)$, then $\cA'$ solves the testing problem in \Cref{def:lin-regr-oblivious} with probability at least $1 - 2\delta$.
\end{theorem}
\begin{proof}
We will argue the success probabilities separately.

\paragraph{Alternate Distribution}
Consider the case when the underlying distribution is alternate and let the latent hidden direction be $v$.
Conditioned on $v$ and $\bU$,
the samples $S$ and $S''$ satisfy the conditions of \Cref{def:estimation-problem} with the underlying regressor $\beta$ and $\bU \beta$, where $\beta:= \rho v$; here we use that Gaussian distribution is rotationally invariant and $\bU x$ is again distributed as isotropic Gaussian.
Thus, the guarantees of $\cA$ imply that with probability $1 - 2 \delta$,
we have that $\|\widehat{\beta}_1 - \beta\|_2 \leq \tau$
and $\|\widehat{\beta}_2 - \beta\|_2 = \|w - \bU \beta\|_2 \leq \tau$.
Since $\tau \leq \rho/4$, we have that $\widehat{\beta}_1\|_2  \leq 1.5 \rho $ and the same for $\|\widehat{\beta}_2\|$.
Since
\begin{align*}
\left \langle \widehat{\beta}_1, \widehat{\beta}_2  \right\rangle - \langle \beta, \beta\rangle
&= - \left \langle \widehat{\beta}_1 - \beta, \widehat{\beta}_2 - \beta  \right\rangle - \langle \beta, \widehat{\beta}_2 - \beta\rangle - \langle \beta, \widehat{\beta}_1 - \beta\rangle,   
\end{align*}
the closeness guarantee implies that
\begin{align*}
\left|\left \langle \widehat{\beta}_1, \widehat{\beta}_2  \right\rangle - \langle \beta, \beta\rangle
\right| \leq \tau^2 + 2\rho \tau  \leq 3 \rho^2/4\,.
\end{align*}
Hence, with probability $1 -2 \delta$, we have that $|W| \geq (\rho^2/4)/(3 \rho/2)^2 \geq 1/9$, and hence the algorithm would correctly output ``alternate''.

\paragraph{Null Distribution}
We will argue that $w$ is independent of $\bU$.
Indeed, for any $\bU$, the distribution of the samples in $S''$ is i.i.d.\ from $\cN(0,\bI_d) \times R$, where $R$ is the marginal distribution of $y$ (recall that $y$ is independent of $X$ under the null).
Hence, $S''$ and $w$ are independent of $\bU$.
Therefore, $ \frac{\widehat{\beta}_2}{\|\widehat{\beta}_2\|_2}$ is distributed uniformly over the unit sphere (independent of $\beta_1$). 
Hence, $W$ is distributed as the product of two unit vectors, implying that with probability $1- \delta$, $|W| \lesssim \sqrt{\frac{\log(1/\delta)}{d}}$, and hence the algorithm correctly outputs ``null'' for $d$ large enough.
\end{proof}

\section{Inefficient SQ Algorithm with Correct Sample Complexity}
\label{app:inefficient-sq}
In this section, we mention an SQ algorithm that uses $q = \exp(\widetilde{O}(d/\tau\alpha))$ queries from $\VSTAT(m)$ with $m = \Theta(1/\alpha)$ and outputs an estimate $\widetilde{\beta}$ such that $\|\widehat{\beta}-\beta\|_2 \lesssim \tau$. Furthermore, 
this SQ algorithm can be simulated from  ${O}\left(\frac{d\log(1/\alpha)}{\alpha}\right)$ i.i.d.\ samples from distribution $P_{\beta^*,E}$.

\begin{theorem}
Let $\|\beta^*\|_2 \leq 1$ and $\alpha\in(0,1)$ and let the underlying distribution be $P_{\beta^*, E}$ for an unknown $E$ and known $\alpha$. 
There exists an SQ algorithm that uses $q \leq \exp(O(d\log(1/\tau\alpha))$  many queries to $\VSTAT(m)$ for $m \lesssim 1/\alpha$
and outputs an estimate $\widetilde{\beta}$ such that $\|\widehat{\beta}-\beta^*\|\lesssim \tau$.

Furthermore, with high probability, the $\VSTAT(m)$ oracle for this SQ algorithm can be simulated using $m' = \widetilde{O}\left(\frac{d}{\alpha}\right)$ many i.i.d.\ samples from $P_{\beta^*,E}$.
\end{theorem}
\begin{proof}
    Let $\cC$ be a $\tau'$-cover of $\{x: \|x\|_2 \leq 1\}$ with respect to the Euclidean norm for $\tau' = 0.01\tau \alpha$. 
    We know such a cover exists with $\log |\cC| \lesssim d \log(1/\tau')$. 
    Furthermore, let $\beta' \in \cC$ be  $\tau'$-close to $\beta^*$.
    For each $\beta \in \cC$, define the query $f_\beta(x,y) = \1_{|x^\top \beta - y| \leq \tau'}$.

    The SQ algorithm is as follows:
    \begin{mdframed}
    \begin{itemize}
        \item For each $\beta \in \cC$, let $v_\beta \gets \VSTAT(f_\beta,m)$.
        \item Output $\widehat{\beta} = \argmax_{\beta \in \cC} v_\beta$.
    \end{itemize}        
    \end{mdframed}

    \paragraph{Correctness.}
    To show correctness, we shall show that for $\beta$ that is $\tau$-far from $\beta^*$, it must be the case
     that $v_{\beta} < v_{\beta'}$, which would imply that any such $\beta$ can not be the output.
     
    Let us start by analyzing $\E[f_{\beta}]$.
    Let the distribution $E$ be $\alpha \delta_0 + (1-\alpha) E'$ for an arbitrary distribution $E'$, where $\delta_0$ is the point mass at origin.
    Then observe that for $G \sim \cN(0,1)$:
    \begin{align*}
        \E[f_\beta] &= \alpha \E_{x}[\1_{|x^\top (\beta - \beta^*)| \leq \tau'}] + (1-\alpha)  \E_{x,z\sim E'}[\1_{|x^\top (\beta - \beta^*)+ z| \leq \tau'}]\\
        &= \alpha \P(|G| \leq \tau'/\|\beta-\beta^*\|) + (1-\alpha) \P_{G,z\sim E'}\left(|G\cdot \|\beta-\beta^*\|_2 + z|\leq \tau'\right).
    \end{align*}
In particular, for $\beta'$, $\E[f_{\beta'}] \geq 0.5\alpha$ because $\|\beta'-\beta^*\|\leq \tau'$ and $\P(|G|\leq 1) \geq 0.5$.
It can then be checked $\max_{\beta \in \cC} \geq v_{\beta'} \geq \E[f_{\beta'}] - \frac{1}{m} - \sqrt{ \frac{\E[f_{\beta'}]}{m} } $, which is bigger than $\E[f_{\beta'}]/2$ if $m \gtrsim \frac{1}{\E[f_{\beta'}]}$, which is satisfied since $m \geq \frac{1}{\alpha}$ and $\E[q_{\beta}]\geq 0.5 \alpha$.

Now consider any $\beta$ such that $\|\beta - \beta^*\| = r \geq \tau =  100\tau'/\alpha$.
   Then
   \begin{align*}
       \E[f_\beta] &=  \alpha \P(|G| \leq \tau'/r) + (1-\alpha) \P_{G,z\sim E'}\left(|G r + z|\leq \tau'\right)\\
       &\leq \alpha \P(|G| \leq \tau'/r) + (1-\alpha) \max_{z' \in \R}\P_{G}\left(|G r + z'|\leq \tau'\right)\\
       &\leq \frac{\alpha \tau'}{r} + (1-\alpha) \frac{\tau'}{r} \\\
       &\leq \frac{\tau'}{r} \leq 0.01 \alpha\,.
   \end{align*}
Therefore, for any such $\beta$, $v_{\beta} \leq \E[f_{\beta'}] + \frac{1}{m} + \sqrt{ \frac{\E[f_{\beta'}]}{m} } \leq 0.02 \alpha $ if $m \gtrsim 1 /\alpha$.
Therefore, any such $\beta$ can not be $\widehat{\beta}$ and hence $\|\widehat{\beta}-\beta^*\|_2 \leq \tau$.

\paragraph{Simulation with samples}
We implement the $\VSTAT(m)$ oracle by taking a set $S$ of i.i.d.\ samples and returning the empirical mean of $q_{\beta}$ over $S$.
Observe that all of the queries $f_\beta$ are halfspaces and hence have VC Dimension $O(d)$.
For $i \in \{1,\dots, \log(1/\alpha_0)\}$,
let $\cA_i = \{\beta: \E[q_{\beta}] \in [2^i\alpha, 2^{i+1}\alpha ] \cup [1 - 2^{i+1}\alpha, 1 - 2^{i}\alpha]\}$. Let $\cA_0 = \{\beta: \E[q_{\beta}] \in [0, \alpha] \cup [1 - \alpha, 1]\}$.

By uniform concentration \cite[Theorem 13.7]{BouLM13} and \cite[Theorem 12.5]{BouLM13}, if $n \geq \frac{d \log(1/2^{i+1}\alpha)}{2^{i+1}\alpha}$, then with probability $1 - \delta/J$ for $J= \log(1/\alpha)$, 
for all $\beta \in \cA_i$ for $i \in \N \cup \{0\}$, we have
\begin{align*}
    &\left|\E_{S}[f_\beta] - \E_{P_{\beta^*,E}}[q_{\beta}]\right| \\
    &\lesssim \sqrt{2^i\alpha }\sqrt{\frac{d\log(1/2^i\alpha)}{n}} + \sqrt{2^{i}\alpha}\cdot \sqrt{\frac{\log(J/\delta)}{n}} +  \frac{\log(J/\delta)}{n} \\
    &\lesssim \sqrt{\E_{P_{\beta^*,E}}[q_{\beta}] \cdot (1 - \E_{P_{\beta^*,E}}[q_{\beta}])} \cdot \sqrt{\frac{dJ + \log(J/\delta)}{n}} + \sqrt{\alpha} \cdot \sqrt{\frac{dJ + \log(J/\delta)}{n}} +  \frac{\log(J/\delta)}{n}\,, 
\end{align*}
where we use that $\alpha + \E[f_\beta]\cdot (1 - \E[f_\beta]) \gtrsim 2^{i} \alpha$ for all $i\in \N \cup \{0\}$.

By a union bound over $\cA_i$'s, this uniform concentration holds for all $\beta \in \R^d$.
That is, if $n \geq \frac{dJ + \log(J/\delta)}{\alpha}$, then with probability $1- \delta$, for all $\beta\in \R^d $,
we have 
\begin{align*}
    &\left|\E_{S}[f_\beta] - \E_{P_{\beta^*,E}}[q_{\beta}]\right| \\
    &\lesssim \sqrt{\E_{P_{\beta^*,E}}[q_{\beta}] \cdot (1 - \E_{P_{\beta^*,E}}[q_{\beta}])} \cdot \sqrt{\frac{dJ + \log(J/\delta)}{n}} + \sqrt{\alpha} \cdot \sqrt{\frac{dJ + \log(J/\delta)}{n}} +  \frac{\log(J/\delta)}{n}\\
    &\leq \sqrt{\E_{P_{\beta^*,E}}[q_{\beta}] \cdot (1 - \E_{P_{\beta^*,E}}[q_{\beta}])} \cdot \sqrt{\frac{1}{m}} + \frac{1}{m}  \,,     
\end{align*}
if $n \gtrsim  (dJ + \log(J/\delta))\cdot \left(m + \alpha m^2  \right) + m\log(J/\delta)$. 
On this event, we get that the empirical approximation is a $\VSTAT(m/4)$ oracle. 
Since we need $m = \Theta(1/\alpha)$, the required sample complexity for failure probability $\delta$ is at most $\frac{d\log(1/\alpha) + \log(\log(1/\alpha)/\delta}{\alpha}$.

\end{proof}
\section{Efficient SQ Algorithm with Matching Accuracy}
\label{app:efficient-sq}
We now show that there exists an efficient SQ algorithm that solves \Cref{def:estimation-problem} and the hard instance in \Cref{thm:sq-hardness-discrete-gaussian} with polynomially number of $\VSTAT(d/\alpha^2)$ queries.
Let $\beta^*$ be the unknown regressor with $\|\beta^*\|_2 \leq 1$.
In this section, we use $u$ as a shorthand for $(x,y)$.

\begin{theorem}
Let  $\alpha \in (0,1)$ and $\beta^* \in \cB$, where $\cB:= \{\beta: \|\beta\|_2\leq 1\}$.
For any $\eps \in (0,1)$, there is an SQ  algorithm that takes these $\alpha, \eps$ as input, makes $\poly(d)$ number of queries to $\VSTAT\left(\frac{d}{\eps \alpha^2}\right)$ on $P_{\beta^*,E}$, and (iii) computes an estimate $\widehat{\beta} \in \R^d$ such that $\|\widehat{\beta}-\beta^*\|_2 \lesssim \eps \alpha$.       
\end{theorem}
Observe that we do not need $\eps$ to be very small to solve the hard instance of \Cref{def:lin-regr-oblivious}, i.e., we can set $\eps = \rho/\alpha =  \widetilde{\Theta}(1)$ and still solve \Cref{def:lin-regr-oblivious} with polynomial number of queries to $\VSTAT(\widetilde{\Theta}(d/\alpha^2))$. 
\begin{proof}
Define the function $g(x) : \cX \to \{0,1\}$ to be function such that $g(x) = 0$ if and only if $\|x\|_2 \geq L \sqrt{d}$ for $L =  \polylog(d/\alpha \eps)$.
Consider the loss function $\ell(\beta, u):= g(x)\cdot\ell_{\mathrm{Huber}}\left(y- x^\top \beta\right)$; here $\ell_{\mathrm{Huber}}(\cdot)$ is the Huber loss with the gradient $h(z) = z\1_{z\in [-1,1]} + \sign(z)\1_{|z|>1}$.
Consider the averaged loss $\cL(\beta) := \E_{u \sim P_{\beta^*, E}}[\ell(\beta, u)]$.

We claim the following:
\begin{enumerate}
    \item $\cL$ is $\kappa$-strongly convex on $\cB$ for $\kappa = \Theta(\alpha)$.
    \item $\cL$ is $L_1$ smooth (Lipschitz continuous gradient) on $\cB$  for $L_1 =  O(1)$.
    \item For every $z \in \cZ$, the function $\ell(\cdot,z)$ is convex, and it is $L_0$-Lipschitz for $L_0 \lesssim L \sqrt{d}$.
    \item $\beta^*$ is the unique minimizer of $\cL$.
\end{enumerate}

Therefore, we can apply  \cite[Corollary 4.12]{FelGV17} with parameters
$L_0$, $L_1$, and $\kappa$ to find an $\alpha \eps$-close estimate $\widehat{\beta}$
such that $\|\widehat{\beta} - \argmin \cL(\beta)\|_2 \lesssim \alpha \eps$ with $ {O}\left(\frac{d L_1 \log(L_1\mathrm{diam}(\cB)/\alpha\eps)}{\kappa}\right) = {O}\left(\frac{d \log(1/\alpha\eps)}{\alpha}\right)$ many queries to $\VSTAT\left(O\left(\frac{L_0^2}{\alpha\eps \kappa}\right)\right) = \VSTAT\left(\frac{d\cdot \polylog(d/\alpha)}{\alpha^2\eps }\right)$.
We get the desired conclusion by noting that $\beta^*$ uniquely minimizes $\cL(\beta)$.

We now give the details omitted earlier:
    \begin{enumerate}
        \item For any unit vector $v$,  $v^\top \nabla^2 \cL v$ is equal to $\E_u[g(x) \nabla^2 \ell_{\mathrm{Huber}} (y - x^\top \beta)(x^\top u)^2 ]$. The convexity follows by non-negativity of the Huber loss.

        \begin{align*}
            v^\top \nabla^2 \cL v  &= \E_u[g(x) \nabla^2 \ell_{\mathrm{Huber}} (y - x^\top \beta)(x^\top v)^2 ] = \E_u[g(u) \I_{|y- x^\top \beta| \leq 1}(x^\top v)^2 ] \\
            &\geq \alpha \cdot \E_{x \sim \cN(0, \bI_d)}[g(x) \I_{|x^\top \beta^*- x^\top \beta| \leq 1}(x^\top v)^2 ] \gtrsim \alpha\,.
        \end{align*}
        The last inequality follows because $g(x) \I_{|x^\top \beta^*- x^\top \beta| \leq 1}(x^\top v)^2 \gtrsim \1_{\|x\|_2 \leq L\sqrt{d}}\1_{|x^\top w |\leq 1}\1_{|x^\top v|\geq 0.5}$ for some unit vector $w$. Using triangle inequality, we obtain that its probability is lower bounded by $\P[\1_{|x^\top w |\leq 1}\1_{|x^\top v|\geq 0.5}] - \P((1-g(x))) \gtrsim 1 - d^{-100} \gtrsim 1$.
        \item The smoothness follows from the same arguments as above by upper bounding $g(x)$ and $\nabla^2 \ell_{\mathrm{Huber}}$ by $1$.
        \item Observe that the gradient satisfies $\nabla \ell(\beta,z) = g(x)h(y - x^\top \beta) x$ and therefore $\|\nabla \ell(\beta,z) x\|_2 \leq L \sqrt{{d}}$, where we use that $\|x g(x)\| \leq \sqrt L {d}$ and the gradient of Huber loss is bounded by $1$.
        
        \item By strong convexity on $\cB$, it suffices to show that $\beta^*$ has zero gradient.
        \begin{align*}
            \|\nabla \cL(\beta^*)\|_2 =
            \| \nabla \E [g(x) h(z)  x]\|_2 = 0,
        \end{align*}
        where we use that $x g(x)$ is a symmetric random variable and independent of $z$.
    \end{enumerate}

 \end{proof}